\documentclass[a4paper]{article}
\usepackage{enumitem}
\usepackage{calc}

\usepackage{lmodern} 

\usepackage{graphicx} 

\usepackage{amsmath, accents}
\usepackage{amssymb,tikz}
\usepackage{pict2e}
\usepackage{amsthm}
\usepackage{amscd}
\usepackage{mathrsfs}
\usepackage{todonotes}
\usetikzlibrary{calc} 

\usepackage{yfonts}

\usepackage{marvosym}
\usepackage[percent]{overpic}

\newtheorem{theorem}{Theorem} [section]
	
\newtheorem{corollary}[theorem]{Corollary}	
\newtheorem{lemma}[theorem]{Lemma}

\newtheorem{remark}[theorem]{Remark}
\theoremstyle{definition}

\newcommand{\re}{\text{\upshape Re\,}}

\let\oldbibliography\thebibliography
\renewcommand{\thebibliography}[1]{\oldbibliography{#1}
\setlength{\itemsep}{-0.5pt}}

\def\XXint#1#2#3{{\setbox0=\hbox{$#1{#2#3}{\int}$}
\vcenter{\hbox{$#2#3$}}\kern-.5\wd0}}

\usepackage{tikz}
\usetikzlibrary{arrows}
\usetikzlibrary{decorations.pathmorphing}
\usetikzlibrary{decorations.markings}
\usetikzlibrary{patterns}
\usetikzlibrary{automata}
\usetikzlibrary{positioning}
\usepackage{tikz-cd}
\tikzset{->-/.style={decoration={
				markings,
				mark=at position #1 with {\arrow{latex}}},postaction={decorate}}}
	
	\tikzset{-<-/.style={decoration={
				markings,
				mark=at position #1 with {\arrowreversed{latex}}},postaction={decorate}}}

\usetikzlibrary{shapes.misc}\tikzset{cross/.style={cross out, draw, 
         minimum size=2*(#1-\pgflinewidth), 
         inner sep=0pt, outer sep=0pt}}

\usepackage{pgfplots}
	
\allowdisplaybreaks	

\numberwithin{equation}{section}

\def\bigO{{\cal O}}

\usepackage[colorlinks=true]{hyperref}
\hypersetup{urlcolor=blue, citecolor=red, linkcolor=blue}

\begin{document}
\title{\vspace*{-1.5cm} Asymptotics of determinants with a rotation-invariant \\ weight and discontinuities along circles}
\author{Christophe Charlier\footnote{Department of Mathematical Sciences, UCPH University of Copenhagen, Universitetsparken 5, 2100 Copenhagen, Denmark. e-mail: charlier@math.ku.dk
}}

\maketitle

\begin{abstract}
We study the moment generating function of the disk counting statistics of a two-dimensional determinantal point process which generalizes the complex Ginibre point process. This moment generating function involves an $n \times n$ determinant whose weight is supported on the whole complex plane, is rotation-invariant, and has discontinuities along circles centered at $0$. These discontinuities can be thought of as a two-dimensional analogue of jump-type Fisher-Hartwig singularities. In this paper, we obtain large $n$ asymptotics for this determinant, up to and including the term of order $n^{-\smash{\frac{1}{2}}}$. We allow for any finite number of discontinuities in the bulk, one discontinuity at the edge, and any finite number of discontinuities bounded away from the bulk. As an application, we obtain the large $n$ asymptotics of all the cumulants of the disk counting function up to and including the term of order $n^{-\smash{\frac{1}{2}}}$, both in the bulk and at the edge. This improves on the best known results for the complex Ginibre point process, and for general values of our parameters these results are completely new. Our proof makes a novel use of the uniform asymptotics of the incomplete gamma function.
\end{abstract}


\section{Introduction and statement of results}\label{section: introduction}
Consider the determinant 
\begin{align*}
D_{n} & :=\frac{1}{n!} \int_{\mathbb{C}}\ldots \int_{\mathbb{C}} \prod_{1 \leq j < k \leq n} |z_{k} -z_{j}|^{2} \prod_{j=1}^{n}w(z_{j})d^{2}z_{j} = \det \left( \int_{\mathbb{C}} z^{j} \overline{z}^{k} w(z) d^{2}z \right)_{j,k=0}^{n-1},
\end{align*}
where the weight $w$ is given by
\begin{align}\label{def of w}
w(z):=|z|^{2\alpha} e^{-n |z|^{2b}} \omega(|z|), \qquad b>0, \; \alpha > -1,
\end{align}
and the function $\omega$ is defined by
\begin{align}\label{def of omega}
\omega(x) = \prod_{\ell=1}^{p} \begin{cases}
e^{u_{\ell}}, & \mbox{if } x < r_{\ell}, \\
1, & \mbox{if } x \geq r_{\ell},
\end{cases}
\end{align}
for some $p \in \mathbb{N}_{>0}:=\{1,2,\ldots\}$, $u_{1},\ldots,u_{p}\in \mathbb{R}$, and $0<r_{1}<\ldots < r_{p} < + \infty$. Thus $w$ is supported on the whole complex plane, is rotation-invariant (i.e. $w(z)=w(|z|)$), has a root-type singularity at $0$, and has $p$ discontinuities along circles centered at $0$. We are interested in the asymptotics of $D_{n}$ as $n \to + \infty$.

\medskip A large number of problems arising in statistical mechanics, integrable operators, orthogonal polynomials, random matrix theory and the theory of Gaussian multiplicative chaos  can be expressed in terms of structured determinants associated with a weight having root-type and/or jump-type singularities (the so-called Fisher-Hartwig singularities) \cite{BottcherReview, DIKreview, Webb, BK2021}. Of particular interest is the asymptotic behavior of these determinants as the size of the underlying matrix gets large. In the case where the weight is supported on a one-dimensional set, these asymptotics have already been widely studied and have a long history. Early works include \cite{Lenard} by Lenard, \cite{Widom2} by Widom, \cite{Basor} by Basor, \cite{BS1985} by B\"ottcher and Silbermann, and \cite{Ehr2001} by Ehrhardt; see also \cite{DIKreview} for more historical background. More recent results have been obtained in e.g. \cite{DIK, DeiftItsKrasovsky} for Toeplitz determinants, \cite{BDIK2015} for Fredholm determinants, \cite{Krasovsky, ItsKrasovsky, BerWebbWong, Charlier, CharlierGharakhloo} for Hankel determinants, \cite{DIK,BasorEhrhardt1} for Toeplitz+Hankel determinants and \cite{CharlierMB} for Muttalib-Borodin determinants. Much less is known when the weight is supported on a two-dimensional set. The first result in this direction is due to Webb and Wong \cite{WebbWong}, who obtained the large $n$ asymptotics of $\det ( \int_{\mathbb{C}} z^{j}\overline{z}^{k}|z-z_{0}|^{2\alpha}e^{-n|z|^{2}} d^{2}z)_{j,k=0}^{n-1}$ with $|z_{0}|<1$ fixed and $\re \alpha > -1$, i.e. they considered the case of a Gaussian weight perturbed with a planar root-type singularity located in the bulk.  Dea\~{n}o and Simm in \cite{DeanoSimm} then investigated the ``edge regime" when $n \to +\infty$ and simultaneously $|z_{0}|\to 1$ at a critical speed. The case of two merging planar root-type singularities in the bulk was also studied in \cite{DeanoSimm}, among other things.  We also mention that the related topic of planar orthogonal polynomials associated with a Gaussian weight having root-type singularities has been studied in \cite{BBLM2015, BGM17, LeeYang, BEG18, LeeYang2, LeeYang3}, see also \cite{AKS2018}. 

\medskip Only limited results are available on asymptotics of determinants with planar discontinuities. In \cite{CE2020}, second order asymptotics were obtained for Ginibre-type weights with discontinuities along smooth curves. In the rotation-invariant setting, but for more general potentials, second order asymptotics were obtained in \cite{L et al 2019}. More refined asymptotics, including the third term of order 1, were then obtained in \cite{FenzlLambert} for Ginibre-type weights. We discuss these works in more detail in Remarks \ref{remark:physics paper on generating}, \ref{remark:universality} and \ref{remark:cumulants} below. The purpose of this paper is to develop a systematic approach to obtain precise asymptotics of determinants with a discontinuous rotation-invariant weight. 

\medskip When the weight is supported on a one-dimensional set, it is now well-understood that the asymptotics analysis of determinants with jump-type Fisher-Hartwig singularities involves hypergeometric functions, see e.g. \cite{ItsKrasovsky}. Our situation presents an obvious but important difference with earlier works such as \cite{ItsKrasovsky}, namely that in \cite{ItsKrasovsky} the discontinuities are located at several isolated points, while in our two-dimensional setting the discontinuities take place along circles. Here we find that the asymptotic analysis of $D_{n}$ involves the uniform asymptotics of the incomplete $\gamma$ function, see Lemma \ref{lemma: uniform} below. 

\medskip To motivate the study of $D_{n}$, let us consider the probability density function
\begin{align}\label{def of point process}
\frac{1}{n!Z_{n}} \prod_{1 \leq j < k \leq n} |z_{k} -z_{j}|^{2} \prod_{j=1}^{n}|z_{j}|^{2\alpha}e^{-n V(z_{j})}, \qquad V(z)=|z|^{2b}, \; b>0, \; \alpha > -1,
\end{align}
where $Z_{n}$ is the normalization constant and $z_{1},\ldots,z_{n} \in \mathbb{C}$. This is the determinantal point process which characterizes the log-potential Coulomb gas with $n$ particles in the external field $V(z)-\frac{2\alpha}{n}\log|z|$ at the inverse temperature $\beta=2$ \cite{ForresterBible}. The density \eqref{def of point process} is also the law of the eigenvalues of an $n \times n$ normal matrix $M$ taken with respect to the probability measure \cite{Mehta}
\begin{align*}
\frac{1}{\mathcal{Z}_{n}}|\det(M)|^{2\alpha}e^{-n \, \mathrm{tr}((MM^{*})^{b})}dM, 
\end{align*}
where $\mathcal{Z}_{n}$ is the normalization constant, $M^{*}$ is the conjugate transpose of $M$, ``$\mbox{tr}$" stands for ``trace" and $dM$ denotes the measure on the set of normal $n \times n$ matrices that is induced by the flat Euclidian metric of $\mathbb{C}^{n\times n}$. The special case $b=1$ and $\alpha=0$ of \eqref{def of point process}, known as the complex Ginibre point process \cite{Ginibre}, is also the joint eigenvalue density of an $n \times n$ random matrix whose entries are independent complex Gaussian random variables with mean $0$ and variance $\frac{1}{n}$. For more background on two-dimensional determinantal point processes such as \eqref{def of point process}, we refer to \cite{HKPV2010}.


\medskip Given a Borel set $A \subset \mathbb{C}$, we denote $N(A) := \#\{z_{j} \in A\}$, i.e. $N(A)$ is the random variable that counts the number of points that fall into $A$. Let $p \in \mathbb{N}_{>0}:=\{1,2,\ldots\}$, let $0<r_{1}<\ldots < r_{p}$, and for $r>0$, let $D_{r} := \{z \in \mathbb{C}: |z|<r\}$. We are interested in the joint moment generating function of $N(D_{r_{1}}),\ldots,N(D_{r_{p}})$, which is given by
\begin{align}
\mathbb{E}\bigg[ \prod_{\ell=1}^{p} e^{u_{\ell}N(D_{r_{\ell}})} \bigg] & = \frac{1}{n!Z_{n}} \int_{\mathbb{C}}\ldots \int_{\mathbb{C}} \prod_{1 \leq j < k \leq n} |z_{k} -z_{j}|^{2} \prod_{j=1}^{n}w(z_{j}) d^{2}z_{j} = \frac{D_{n}}{Z_{n}}, \label{def of Dn as n fold integral} 
\end{align}
where $u_{1},\ldots,u_{p}\in \mathbb{R}$ and the weight $w$ was defined in \eqref{def of w}. We mention that $Z_{n}$ can be easily expanded as $n \to  +\infty$ (see Remark \ref{remark:asymp of Zn} below). Thus the two problems of obtaining the large $n$ asymptotics of $D_{n}$ and of $\mathbb{E}\big[ \prod_{\ell=1}^{p} e^{u_{\ell}N(D_{r_{\ell}})} \big]$ are essentially equivalent.


\medskip The main result of this paper is an explicit formula for the large $n$ asymptotics of $\mathbb{E}\big[ \prod_{\ell=1}^{p} e^{u_{\ell}N(D_{r_{\ell}})} \big]$, up to and including the term of order $\smash{n^{-\frac{1}{2}}}$. These asymptotics depend very much on whether the $r_{\ell}$'s are smaller, equal, or bigger than the critical value $b^{-\frac{1}{2b}}$. This is because the normalized empirical distribution $\frac{1}{n}\sum_{j=1}^{n}\delta_{z_{j}}$ of \eqref{def of point process} converges as $n \to +\infty$ weakly almost surely (see e.g. \cite{KS1999, CGZ2014}) to an equilibrium measure $\mu$ \cite{SaTo},
\begin{align}\label{equilibrium measure}
d\mu(z) = \frac{1}{4\pi}\Delta V(z)d^{2}z = \frac{b^{2}}{\pi}|z|^{2b-2}d^{2}z,
\end{align}
which is supported on the closed disk centered at $0$ of radius $\smash{b^{-\frac{1}{2b}}}$. In Theorem \ref{thm:main thm} below, we treat the general situation where
\begin{align*}
0 < r_{1} < \ldots <r_{m} < r_{m+1}=b^{-\frac{1}{2b}}\bigg( 1+\sqrt{2b}\frac{\mathfrak{s}}{\sqrt{n}} \bigg)^{\frac{1}{2b}} < r_{m+2} < \ldots < r_{p} < +\infty, \qquad \mathfrak{s}\in \mathbb{R},
\end{align*}
i.e. $w$ has $m$ discontinuities strictly inside the support of $\mu$ (the bulk), one discontinuity close to the edge, and $p-m-1$ discontinuities outside the support of $\mu$. 

\medskip The large $n$ asymptotics of $\mathbb{E}\big[ \prod_{\ell=1}^{p} e^{u_{\ell}N(D_{r_{\ell}})} \big]$ are naturally described in terms of the two functions
\begin{align}\label{functions Fb and Gb}
& \mathcal{F}(t,s):= \log \bigg( 1+\frac{s-1}{2}\mathrm{erfc}(t) \bigg), \qquad \mathcal{G}(t,s):= \frac{1-s}{1+\frac{s-1}{2}\mathrm{erfc}(t)}\frac{e^{-t^{2}}}{\sqrt{\pi}} = \frac{d}{dt}\mathcal{F}(t,s),
\end{align}
where $t \in \mathbb{R}$, $s\in \mathbb{C}\setminus(-\infty,0]$, the principal branch is chosen for the $\log$, and $\mathrm{erfc}$ is the complementary error function defined by
\begin{align}\label{def of erfc}
\mathrm{erfc} (t) = \frac{2}{\sqrt{\pi}}\int_{t}^{\infty} e^{-x^{2}}dx.
\end{align}
We now state our main result.
\begin{theorem}\label{thm:main thm}
Let $p \in \mathbb{N}_{>0}$, $m \in \{0,1,\ldots,p-1\}$, $\mathfrak{s} \in \mathbb{R}$, and
\begin{align*}
\alpha > -1, \qquad b>0, \qquad 0 < r_{1} < \ldots <r_{m} < b^{-\frac{1}{2b}} < r_{m+2} < \ldots < r_{p} < +\infty,
\end{align*}
be fixed parameters, and for $n \in \mathbb{N}_{>0}$, define $r_{m+1} := b^{-\frac{1}{2b}}\big( 1+\sqrt{2b}\frac{\mathfrak{s}}{\sqrt{n}} \big)^{\frac{1}{2b}}$. For any fixed $x_{1},\ldots,x_{p} \in \mathbb{R}$, there exists $\delta > 0$ such that 
\begin{align}\label{asymp in main thm}
\mathbb{E}\bigg[ \prod_{j=1}^{p} e^{u_{j}N(D_{r_{j}})} \bigg] = \exp \bigg( C_{1} n + C_{2} \sqrt{n} + C_{3} +  \frac{C_{4}}{\sqrt{n}} + \bigO\bigg(\frac{(\log n)^{2}}{n}\bigg)\bigg), \qquad \mbox{as } n \to + \infty
\end{align}
uniformly for $u_{1} \in \{z \in \mathbb{C}: |z-x_{1}|\leq \delta\},\ldots,u_{p} \in \{z \in \mathbb{C}: |z-x_{p}|\leq \delta\}$, where
\begin{align*}
& C_{1} = \sum_{j=1}^{m} b r_{j}^{2b} u_{j} + \sum_{j=m+1}^{p} u_{j}, \\
& C_{2} = \sum_{j=1}^{m} \sqrt{2}b r_{j}^{b}\int_{0}^{+\infty} \Big(\mathcal{F}(t,e^{u_{j}}) + \mathcal{F}(t,e^{-u_{j}}) \Big)dt + \sqrt{2b} \int_{0}^{+\infty} \mathcal{F}(t,e^{-u_{m+1}}) dt \\
& + \sqrt{2b} \; \mathfrak{s} \, u_{m+1} + \sqrt{2b} \int_{0}^{-\mathfrak{s}} \mathcal{F}(t,e^{u_{m+1}})dt, \\
& C_{3} = - \bigg( \frac{1}{2}+\alpha \bigg)\sum_{j=1}^{m}u_{j} + \bigg( \frac{1}{2}+\alpha \bigg) \mathcal{F}(\mathfrak{s},e^{-u_{m+1}}) + 4 b \sum_{j=1}^{m} \int_{0}^{+\infty} t \Big( \mathcal{F}(t,e^{u_{j}}) - \mathcal{F}(t,e^{-u_{j}}) \Big)dt  \\
& - 2b \int_{0}^{+\infty} (2t-\mathfrak{s}) \, \mathcal{F}(t,e^{-u_{m+1}}) dt + 2b \int_{0}^{-\mathfrak{s}} (2t+\mathfrak{s}) \, \mathcal{F}(t,e^{u_{m+1}}) dt \\
& + b \sum_{j=1}^{m} \int_{-\infty}^{+\infty} \mathcal{G}(t,e^{u_{j}}) \frac{5t^{2}-1}{3} dt + b  \int_{-\infty}^{-\mathfrak{s}} \mathcal{G}(t,e^{u_{m+1}})\frac{5 t^{2} +3\mathfrak{s} t - 1}{3} dt, \\
& C_{4} =  \sum_{j=1}^{m} \frac{6\sqrt{2}\, b}{r_{j}^{b}}\int_{0}^{+\infty} t^{2}\Big(\mathcal{F}(t,e^{u_{j}}) + \mathcal{F}(t,e^{-u_{j}}) \Big)dt + (2b)^{3/2}\int_{0}^{+\infty} (3t^{2}-2\mathfrak{s}t)  \mathcal{F}(t,e^{-u_{m+1}}) dt \\
& + (2b)^{3/2}\int_{0}^{-\mathfrak{s}} (3t^{2}+2\mathfrak{s}t)  \mathcal{F}(t,e^{u_{m+1}}) dt + \sum_{j=1}^{m} \frac{-b}{\sqrt{2}\, r_{j}^{b}} \int_{-\infty}^{+\infty} \mathcal{G}(t,e^{u_{j}}) \frac{21 t - 193 t^{3} + 50 t^{5}}{18} dt \\
& - \frac{b^{3/2}}{\sqrt{2}} \int_{-\infty}^{-\mathfrak{s}} \mathcal{G}(t,e^{u_{m+1}}) \frac{21 t - 193 t^{3} + 50 t^{5} + 6\mathfrak{s}(1-29t^{2}+10t^{4})-9\mathfrak{s}^{2}(3t-2t^{3})}{18} dt \\
& - \sum_{j=1}^{m} \frac{b}{2\sqrt{2}r_{j}^{b}} \int_{-\infty}^{+\infty} \bigg( \mathcal{G}(t,e^{u_{j}}) \frac{5t^{2}-1}{3}  \bigg)^{2}dt  - \frac{b^{3/2}}{2\sqrt{2}} \int_{-\infty}^{-\mathfrak{s}} \bigg( \mathcal{G}(t,e^{u_{m+1}})  \frac{5t^{2}+3\mathfrak{s}t-1}{3} \bigg)^{2}dt \\
& + \bigg( \bigg(\frac{1}{2}+\alpha\bigg) \frac{2\mathfrak{s}^{2}-1}{3\sqrt{2}}\sqrt{b} + \frac{1+6\alpha+6\alpha^{2}}{12\sqrt{2b}} \bigg) \mathcal{G}(-\mathfrak{s},e^{u_{m+1}}).
\end{align*}
In particular, since $\mathbb{E}\big[ \prod_{j=1}^{p} e^{u_{j}N(D_{r_{j}})} \big]$ is analytic for $u_{1},\ldots,u_{p} \in \mathbb{C}$ and positive for $u_{1},\ldots,u_{p} \in \mathbb{R}$, Cauchy's formula combined with \eqref{asymp in main thm} implies that for any $k_{1},\ldots,k_{p}\in \mathbb{N}:=\{0,1,\ldots\}$, $k_{1}+\ldots+k_{p}\geq 1$, and $u_{1},\ldots,u_{p}\in \mathbb{R}$, we have
\begin{align}\label{der of main result}
\partial_{u_{1}}^{k_{1}}\ldots \partial_{u_{p}}^{k_{p}} \bigg\{ \log \mathbb{E}\bigg[ \prod_{j=1}^{p} e^{u_{j}N(D_{r_{j}})} \bigg] - \bigg( C_{1} n + C_{2} \sqrt{n} + C_{3} +  \frac{C_{4}}{\sqrt{n}} \bigg) \bigg\} = \bigO\bigg(\frac{(\log n)^{2}}{n}\bigg), \quad \mbox{as } n \to + \infty.
\end{align}
\end{theorem}
\begin{remark}
With more efforts, we expect that the estimates $\bigO\big(\frac{(\log n)^{2}}{n}\big)$ in \eqref{asymp in main thm} and \eqref{der of main result} can actually be shown to be $\bigO(n^{-1})$.
\end{remark}
\begin{remark}\label{remark:physics paper on generating}
For $b=1$, $\alpha=0$, $u_{2}=\ldots=u_{p}=0$ and $r_{1}<1$ (the bulk regime of the complex Ginibre point process), the first two leading coefficients $C_{1}$ and $C_{2}$ were previously obtained in \cite[eq (28)]{L et al 2019}, and $C_{3}$ was obtained in \cite[Proposition 2.1]{FenzlLambert}. For $b=1$, $\alpha=0$, $u_{2}=\ldots=u_{p}=0$ and $r_{1} = \big( 1+\sqrt{2}\frac{\mathfrak{s}}{\sqrt{n}} \big)^{\frac{1}{2}}$ (the edge regime of the complex Ginibre point process), $C_{1},C_{2},C_{3}$ were also obtained in \cite[Proposition 2.1]{FenzlLambert}.
\end{remark}

\begin{remark}\label{remark:universality}
Consider the probability measure \eqref{def of point process} with a general radial potential $V(z)=V(|z|)$ satisfying $V(|z|)/(2\log |z|)\to + \infty$ as $|z|\to + \infty$, and assume that the equilibrium measure is supported on $D_{r_{\star}}$ for a certain $r_{\star}>0$. In this general setting, it was argued in \cite{L et al 2019} that for $r<r_{\star}$ fixed, we have $\mathbb{E}[e^{uN(D_{r})}]=\exp(C_{1}n+C_{2}\sqrt{n}+o(\sqrt{n}))$ as $n \to + \infty$ for certain coefficients $C_{1}$ and $C_{2}$. It was also conjectured in \cite[eq (34)]{L et al 2019} that $C_{2}$ can be written in the form $C_{2}=c_{V}\widetilde{C}_{2}(u)$, where $c_{V}$ is independent of $u$, and $\widetilde{C}_{2}(u)$ is a universal quantity independent of $V$. Theorem \ref{thm:main thm} establishes this conjecture for $V(z)=|z|^{2b}$. More generally, we note from Theorem \ref{thm:main thm} that the coefficients $C_{2},C_{3}$ and $C_{4}$ appearing in the large $n$ asymptotics of $\mathbb{E}\big[ \prod_{j=1}^{m} e^{u_{j}N(D_{r_{j}})} \big]$ are of the forms
\begin{align*}
C_{2} = \sum_{j=1}^{m}br_{j}^{b}\widetilde{C}_{2}(u_{j}), & & C_{3} = - \bigg( \frac{1}{2}+\alpha \bigg)\sum_{j=1}^{m}u_{j} + \sum_{j=1}^{m}b\widetilde{C}_{3}(u_{j}), & & C_{4} = \sum_{j=1}^{m}br_{j}^{-b}\widetilde{C}_{4}(u_{j}), 
\end{align*}
for some explicit $\{\widetilde{C}_{k}(u)\}_{k=2}^{4}$ that are independent of $\alpha, b$ and $r_{1},\ldots,r_{m}$. We also find it remarkable that $C_{3}$ is completely independent of $r_{1},\ldots,r_{m}$.
\end{remark}
\begin{remark}
For one-dimensional $\log$-correlated point processes, asymptotic formulas for moment generating functions of bulk counting statistics are typically of the form $\exp(D_{1}n + D_{2}\log n + D_{3} + o(1))$, see e.g. \cite{Charlier}, and thus differ drastically from the asymptotics \eqref{asymp in main thm}.
\end{remark}
Recall that the cumulants $\{\kappa_{j}=\kappa_{j}(n,r,b,\alpha) \}_{j\in \mathbb{N}_{>0}}$ of the random variable $N(D_{r})$ are defined through the expansion
\begin{align*}
\log \mathbb{E}[e^{uN(D_{r})}] = \kappa_{1} u +  \frac{\kappa_{2}u^{2}}{2!} +  \frac{\kappa_{3}u^{3}}{3!} + \frac{\kappa_{4}u^{4}}{4!} + \ldots, \qquad \mbox{as } u \to 0,
\end{align*}
or equivalently by
\begin{align}\label{cumulant}
\kappa_{j} = \partial_{u}^{j} \log \mathbb{E}[e^{uN(D_{r})}] \Big|_{u=0}.
\end{align}
More generally, the joint cumulants of $N(D_{r_{1}}), \ldots, N(D_{r_{p}})$ are defined by
\begin{align}\label{joint cumulant}
\kappa_{j_{1},\ldots,j_{p}}:=\partial_{u_{1}}^{j_{1}}\ldots \partial_{u_{p}}^{j_{p}} \log \mathbb{E}[e^{u_{1}N(D_{r_{1}})+\ldots + u_{p}N(D_{r_{p}})}] \Big|_{u_{1}=\ldots=u_{p}=0}, \qquad j_{1},\ldots,j_{p}\in \mathbb{N}_{>0}.
\end{align}
We can deduce from Theorem \ref{thm:main thm} the following results.
\begin{corollary}\label{coro:counting function} \,
\begin{itemize}
\item[(a)] (Asymptotics for the cumulants in the bulk regime) \\[0.2cm]
Let $j \in \mathbb{N}_{>0}$, $\alpha > -1$, $b>0$ and $r \in (0,b^{-\frac{1}{2b}})$ be fixed. As $n \to +\infty$, we have
\begin{align}\label{formula cumulant bulk}
\kappa_{j}  = \begin{cases}
b r^{2b} n + d_{j} + \bigO\big(\frac{(\log n)^{2}}{n}\big), & \mbox{if } j=1, \\[0.2cm]
\hspace{1.25cm} d_{j} + \bigO\big(\frac{(\log n)^{2}}{n}\big), & \mbox{if } $j$ \mbox{ is odd and } j\neq 1, \\[0.2cm]
c_{j} \sqrt{n} + e_{j} n^{-\frac{1}{2}} + \bigO\big(\frac{(\log n)^{2}}{n}\big), & \mbox{if } $j$ \mbox{ is even},
\end{cases}
\end{align}
where
\begin{align}
& c_{j} = \sqrt{2} \, b r^{b} \int_{0}^{+\infty} \partial_{u}^{j}\Big(\mathcal{F}(t,e^{u}) + \mathcal{F}(t,e^{-u}) \Big)\Big|_{u=0}dt, \nonumber \\
& d_{j} = \left\{ \begin{array}{c c}
-\frac{1}{2}-\alpha, & \mbox{if } j=1 \\
0, & \mbox{if } j \geq 2
\end{array} \right\} + 4b \int_{0}^{+\infty} t \, \partial_{u}^{j}\Big( \mathcal{F}(t,e^{u}) - \mathcal{F}(t,e^{-u}) \Big)\Big|_{u=0}dt \nonumber \\
& \hspace{+0.7cm} + b \int_{-\infty}^{+\infty} \partial_{u}^{j}\mathcal{G}(t,e^{u})\big|_{u=0} \frac{5t^{2}-1}{3} dt, \nonumber \\
& e_{j} = \frac{6\sqrt{2} \, b}{r^{b}} \int_{0}^{+\infty} t^{2}\partial_{u}^{j}\Big(\mathcal{F}(t,e^{u}) + \mathcal{F}(t,e^{-u}) \Big)\Big|_{u=0}dt - \frac{b}{ r^{b}} \int_{-\infty}^{+\infty} \partial_{u}^{j}\mathcal{G}(t,e^{u})\big|_{u=0} \frac{21 t - 193 t^{3} + 50 t^{5}}{18\sqrt{2}} dt \nonumber \\
& \hspace{+0.8cm} - \frac{b}{2\sqrt{2} \, r^{b}} \int_{-\infty}^{+\infty} \partial_{u}^{j} \big[\mathcal{G}(t,e^{u})^{2}\big]\big|_{u=0} \bigg( \frac{5t^{2}-1}{3} \bigg)^{2}dt. \label{cjdjej bulk}
\end{align}
For $j=1$ and $j=2$, these integrals can be simplified further, and we obtain 
\begin{align}
& \kappa_{1} = \mathbb{E}[N(D_{r})] = b r^{2b} n + \frac{b-1-2\alpha}{2} + \bigO\bigg(\frac{(\log n)^{2}}{n}\bigg), \label{mean bulk} \\
& \kappa_{2} = \mathrm{Var}[N(D_{r})] = \frac{b r^{b}}{\sqrt{\pi}} \sqrt{n} - \frac{b}{16 \sqrt{\pi}r^{b}}\frac{1}{\sqrt{n}} + \bigO\bigg(\frac{(\log n)^{2}}{n}\bigg), \label{var bulk}
\end{align}
as $n \to +\infty$. As one would expect, the leading term of $\mathbb{E}[N(D_{r})]$ in \eqref{mean bulk} can be rewritten as $br^{2b}=\int_{D_{r}}d\mu$, where $\mu$ is the equilibrium measure defined in \eqref{equilibrium measure}.

\item[(b)] (Asymptotics for the cumulants in the edge regime) \\[0.1cm]
Let $j \in \mathbb{N}_{>0}$, $\alpha > -1$, $b>0$ and $\mathfrak{s}\in \mathbb{R}$ be fixed, and for $n \in \mathbb{N}_{>0}$, let $r:=b^{-\frac{1}{2b}}\big( 1+\sqrt{2b}\frac{\mathfrak{s}}{\sqrt{n}} \big)^{\frac{1}{2b}}$. As $n \to +\infty$, we have
\begin{align}\label{formula cumulant edge}
\kappa_{j}  = \begin{cases}
n + c_{j} \sqrt{n} + d_{j} + e_{j} n^{-\frac{1}{2}} + \bigO\big(\frac{(\log n)^{2}}{n}\big), & \mbox{if } j=1, \\[0.2cm]
\hspace{0.65cm} c_{j} \sqrt{n} + d_{j} + e_{j} n^{-\frac{1}{2}} + \bigO\big(\frac{(\log n)^{2}}{n}\big), & \mbox{if } j\geq 2,
\end{cases}
\end{align}
where
\begin{align*}
& c_{j} = \left\{\begin{array}{c c} 
\sqrt{2b} \, \mathfrak{s}, & \mbox{if } j=1 \\
0, & \mbox{if } j \geq 2
\end{array}
\right\} + \sqrt{2b}\int_{0}^{+\infty} \partial_{u}^{j}\mathcal{F}(t,e^{-u})\big|_{u=0} dt + \sqrt{2b}\int_{0}^{-\mathfrak{s}} \partial_{u}^{j}\mathcal{F}(t,e^{u})\big|_{u=0}dt, \\
& d_{j} = \bigg( \frac{1}{2}+\alpha \bigg) \partial_{u}^{j}\mathcal{F}(\mathfrak{s},e^{-u})\big|_{u=0} - 2b\int_{0}^{+\infty} (2t-\mathfrak{s}) \, \partial_{u}^{j}\mathcal{F}(t,e^{-u})\big|_{u=0} dt \\
& + 2b\int_{0}^{-\mathfrak{s}} (2t+\mathfrak{s}) \, \partial_{u}^{j}\mathcal{F}(t,e^{u})\big|_{u=0} dt +  b \int_{-\infty}^{-\mathfrak{s}} \partial_{u}^{j}\mathcal{G}(t,e^{u})\big|_{u=0}  \frac{5t^{2}+3\mathfrak{s}t-1}{3} dt, \\
& e_{j} = (2b)^{3/2}\int_{0}^{+\infty} (3t^{2}-2\mathfrak{s}t)  \partial_{u}^{j}\mathcal{F}(t,e^{-u})\big|_{u=0} dt + (2b)^{3/2}\int_{0}^{-\mathfrak{s}} (3t^{2}+2\mathfrak{s}t) \partial_{u}^{j}\mathcal{F}(t,e^{u})\big|_{u=0} dt \\
& - \frac{b^{3/2}}{\sqrt{2}}\int_{-\infty}^{-\mathfrak{s}} \partial_{u}^{j}\mathcal{G}(t,e^{u})\big|_{u=0} \frac{21 t - 193 t^{3} + 50 t^{5} + 6\mathfrak{s}(1-29t^{2}+10t^{4})-9\mathfrak{s}^{2}(3t-2t^{3})}{18} dt \\
& - \frac{b^{3/2}}{2\sqrt{2}} \int_{-\infty}^{-\mathfrak{s}} \partial_{u}^{j}\big[\mathcal{G}(t,e^{u})^{2}\big]\big|_{u=0} \bigg( \frac{5 t^{2} + 3\mathfrak{s} t-1}{3} \bigg)^{2}dt \\
& + \bigg( \bigg(\frac{1}{2}+\alpha\bigg) \frac{2\mathfrak{s}^{2}-1}{3\sqrt{2}}\sqrt{b} + \frac{1+6\alpha+6\alpha^{2}}{12\sqrt{2b}} \bigg) \partial_{u}^{j} \mathcal{G}(-\mathfrak{s},e^{u})\big|_{u=0}.
\end{align*}
For $j=1$ and $j=2$, the coefficients $c_{j}$, $d_{j}$ and $e_{j}$ can be evaluated explicitly (in terms of $\mathrm{erfc}$) using integration by parts, and we obtain the following:
\begin{align}
& c_{1} = \frac{\sqrt{b} \, \mathfrak{s}}{\sqrt{2}}\mathrm{erfc}(\mathfrak{s}) - \frac{\sqrt{b}}{\sqrt{2\pi}}e^{-\mathfrak{s}^{2}}, \nonumber \\
& d_{1} = -\frac{1}{2}\bigg( \frac{1}{2}+\alpha - \frac{b}{2} \bigg) \mathrm{erfc} (\mathfrak{s}) - \frac{b \, \mathfrak{s}}{3\sqrt{\pi}}e^{-\mathfrak{s}^{2}}, \nonumber \\
& e_{1} = \frac{e^{-\mathfrak{s}^{2}}}{\sqrt{2\pi}}\bigg( \frac{b(2+4\alpha)-1-6\alpha-6\alpha^{2}}{12\sqrt{b}} + \frac{(3b-2-4\alpha)\mathfrak{s}^{2}}{6}\sqrt{b} - \frac{2\mathfrak{s}^{4}}{9}b^{3/2} \bigg), \nonumber \\
& c_{2} = \frac{\sqrt{b}}{2\sqrt{\pi}}\mathrm{erfc}(\sqrt{2} \, \mathfrak{s}) + \sqrt{b}\frac{e^{-\mathfrak{s}^{2}}}{\sqrt{2\pi}}\big( 1-\mathrm{erfc} (\mathfrak{s})\big) + \frac{\sqrt{b} \, \mathfrak{s}}{\sqrt{2}}\mathrm{erfc} ( \mathfrak{s}) \bigg( \frac{1}{2}\mathrm{erfc}( \mathfrak{s})-1 \bigg), \nonumber \\
& d_{2} = -\frac{b}{12 \pi}e^{-2 \mathfrak{s}^{2}} + \frac{b \,\mathfrak{s}}{2\sqrt{2\pi}}\mathrm{erfc}( \sqrt{2} \,  \mathfrak{s} ) + \frac{b \, \mathfrak{s}}{3\sqrt{\pi}}e^{-\mathfrak{s}^{2}}\big( 1-\mathrm{erfc}( \mathfrak{s}) \big) \nonumber \\
& \hspace{0.75cm} + \frac{b-1-2\alpha}{4} \mathrm{erfc}( \mathfrak{s}) \bigg( \frac{1}{2}\mathrm{erfc}( \mathfrak{s}) - 1 \bigg), \nonumber \\
& e_{2} = \frac{e^{-\mathfrak{s}^{2}}}{12\sqrt{2\pi b}} \bigg( 1-2b+6\alpha-4b \alpha + 6\alpha^{2} + 2(2-3b+4\alpha) b \,\mathfrak{s}^{2} + \frac{8b^{2}}{3}\mathfrak{s}^{4} \bigg) \big( 1-\mathrm{erfc}( \mathfrak{s} ) \big) \nonumber \\
& \hspace{0.75cm}  -\frac{b^{3/2}  \mathfrak{s}}{72 \sqrt{2} \, \pi}e^{-2\mathfrak{s}^{2}} - \frac{b^{3/2}(1+4\mathfrak{s}^{2})}{32\sqrt{\pi}}\mathrm{erfc}( \sqrt{2} \, \mathfrak{s} ). \label{coeff for exp and var}
\end{align}
In particular, for $r=b^{-\frac{1}{2b}}$ (thus $D_{r}=\mathrm{supp} \, \mu$ and $\mathfrak{s}=0$), as $n \to + \infty$ we have
\begin{align*}
& \mathbb{E}[N(D_{r})] = n - \frac{\sqrt{b}}{\sqrt{2 \pi}}\sqrt{n} + \frac{b-1-2\alpha}{4} + \frac{b(2+4\alpha)-1-6\alpha - 6\alpha^{2}}{12\sqrt{2\pi b}\sqrt{n}} + \bigO\bigg(\frac{(\log n)^{2}}{n}\bigg), \\
& \mathrm{Var}[N(D_{r})] = \frac{\sqrt{b}}{2\sqrt{\pi}} \sqrt{n} + \frac{1+2\alpha-b}{8} - \frac{b}{12\pi} - \frac{b^{3/2}}{32 \sqrt{\pi}}\frac{1}{\sqrt{n}} + \bigO\bigg(\frac{(\log n)^{2}}{n}\bigg).
\end{align*}
\item[(c)] (Asymptotics for the cumulants in the regime bounded away from the bulk) \\[0.2cm] 
Let $\alpha > -1$, $b>0$ and $r > b^{-\frac{1}{2b}}$ be fixed. As $n \to +\infty$, we have
\begin{align}\label{formulas cumulants outside the bulk}
\kappa_{j}  = \begin{cases}
n + \bigO\big(\frac{(\log n)^{2}}{n}\big), & \mbox{if } j=1, \\[0.2cm]
\hspace{0.65cm} \bigO\big(\frac{(\log n)^{2}}{n}\big), & \mbox{if } j\geq 2.
\end{cases}
\end{align}
\item[(d)] (Asymptotics for the joint cumulants) \\[0.2cm] 
Let $p \in \mathbb{N}$, $p \geq 2$, $m\in \{0,1,\ldots,p-1\}$, $j_{1},\ldots,j_{p} \in \mathbb{N}$, $\alpha > -1$, $b>0$, and 
\begin{align*}
0 < r_{1}<r_{2}<\ldots<r_{m}<b^{-\frac{1}{2b}} < r_{m+2}<\ldots < r_{p}<+\infty, \qquad \mathfrak{s}\in \mathbb{R},
\end{align*}
be fixed parameters, and for $n \in \mathbb{N}_{>0}$, define $r_{m+1}=b^{-\frac{1}{2b}}(1+\sqrt{2b}\frac{\mathfrak{s}}{\sqrt{n}})^{\frac{1}{2b}}$. If at least two $j_{\ell}$'s are positive, then as $n \to +\infty$ we have
\begin{align}\label{cov asymp}
& \kappa_{j_{1},\ldots,j_{p}} = \bigO\bigg(\frac{(\log n)^{2}}{n}\bigg).
\end{align}
\item[(e)] (joint Gaussian fluctuations) \\[0.2cm] 
Let $\alpha > -1$, $b>0$, $\mathfrak{s}\in \mathbb{R}$ and $0<r_{1}<r_{2}<\ldots<r_{m}<b^{-\frac{1}{2b}}$ be fixed, and for $n \in \mathbb{N}_{>0}$, define $r_{m+1}:=b^{-\frac{1}{2b}}\big( 1+\sqrt{2b}\frac{\mathfrak{s}}{\sqrt{n}} \big)^{\frac{1}{2b}}$. Consider the random variables
\begin{align}
& N_{j} := \pi^{1/4}\frac{N(D_{r_{j}})-br_{j}^{2b}n}{\sqrt{br_{\smash{j}}^{b}} \; n^{1/4}}, \qquad j=1,\ldots,m, \label{Nj bulk} \\
& N_{m+1}:= \frac{N(D_{r_{m+1}})-(n+c_{1}\sqrt{n})}{\sqrt{c_{2}}\; n^{1/4}}, \label{Nj edge}
\end{align}
where $c_{1},c_{2}$ are as in \eqref{coeff for exp and var}. As $n \to + \infty$, $(N_{1},\ldots,N_{m+1})$ convergences in distribution to a multivariate normal random variable of mean $(0,\ldots,0)$ and covariance matrix $I_{m+1}$, where $I_{m+1}$ is the $(m+1) \times (m+1)$ identity matrix.
\end{itemize}
\end{corollary}
\begin{remark}\label{remark:cumulants}
Some parts of Corollary \ref{coro:counting function} were already known:
\begin{itemize}
\item In \cite{Rider}, Rider obtained various results for the variance and covariance of radial and angular statistics in the complex Ginibre point process (which corresponds to $b=1$ and $\alpha=0$ in our setting). In particular, for $(b,\alpha)=(1,0)$, the leading coefficient $c_{2}=b r^{b}/\sqrt{\pi}$ of \eqref{var bulk} was determined in \cite[Theorem 1.6]{Rider}.
\item Let $N_{\mathrm{Ell}}$ be the number of points lying outside the droplet of the Elliptic Ginibre ensemble. Fine asymptotics for $\mathbb{E}[N_{\mathrm{Ell}}]$, including the term of order $n^{-\frac{1}{2}}$, were obtained in \cite[eq (70)]{LeeRiser2016}. In particular, for the edge regime, the coefficients $c_{1}|_{(\mathfrak{s},b,\alpha)=(0,1,0)}$, $d_{1}|_{(\mathfrak{s},b,\alpha)=(0,1,0)}$ and $e_{1}|_{(\mathfrak{s},b,\alpha)=(0,1,0)}$ were previously found in \cite{LeeRiser2016}.
\item Let $N_{A} := \#\{z_{j}:z_{j}\in A\}$ be the number of points of the complex Ginibre process lying in a given Borel set $A$. The following was proved in \cite[Theorem 1.6]{CE2020}: when $A$ is in the bulk, has smooth boundary, and is independent of $n$, the cumulants $\{\kappa_{j}(A)\}_{j=1}^{+\infty}$ satisfy
\begin{align}\label{all order expansion}
\kappa_{j}(A) = \begin{cases}
\alpha_{j,0}n \hspace{+0.2cm} + \sum_{\ell=1}^{N} \alpha_{j,\ell}n^{1-\ell} + \bigO(n^{-N}), & \mbox{if } j =1, \\
\hspace{1.43cm} \sum_{\ell=1}^{N} \alpha_{j,\ell}n^{1-\ell} + \bigO(n^{-N}), & \mbox{if $j$ is odd and } j \geq 3, \\
\beta_{j,0}n^{\frac{1}{2}} + \sum_{\ell=1}^{N} \beta_{j,\ell}n^{\frac{1}{2}-\ell} + \bigO(n^{-N-\frac{1}{2}}), & \mbox{if $j$ is even,}
\end{cases}
\end{align}
for any $N \in \mathbb{N}$ and some $\alpha_{j,\ell},\beta_{j,\ell} \in \mathbb{R}$. Moreover, the constants $\alpha_{1,0}, \beta_{j,0}$ were also determined explicitly in \cite[Theorem 1.6]{CE2020}. Our asymptotics \eqref{coro:counting function} are consistent with \eqref{all order expansion}. (Actually, our results also suggest that the cumulants $\kappa_{j}(D_{r})$ of the Mittag-Leffler ensemble with general $b$ and $\alpha$ also satisfy an all-order expansion of the form \eqref{all order expansion}.)
\item The coefficients $\{c_{j}\}_{j=1}^{+\infty}$ were obtained for both the bulk and the edge regimes of the complex Ginibre point process in \cite[eqs (55)--(67)]{LMS2018} (the analysis of \cite{LMS2018} is done in the context of fermions in a rotating trap, and this model is equivalent to the complex Ginibre point process \cite{LMS2018}, see also \cite{KMS2021}). The coefficients $\{d_{j}|_{(b,\alpha)=(1,0)}\}_{j=1}^{+\infty}$ were then obtained in \cite[Remark 4]{FenzlLambert} for the bulk regime.
\item Corollary \ref{coro:counting function} (e), when specialized to $(b,\alpha)=(1,0)$, was already known from \cite{Rider} for $m=1$, and from \cite[Proposition 2.2]{FenzlLambert} for general $m \in \mathbb{N}_{>0}$. Note that disk counting statistics are linear statistics with indicator type test functions (thus non-smooth). We mention in passing that for smooth linear statistics of non-Hermitian random matrices, some Gaussian fluctuation formulas were already obtained by Forrester in \cite{Forrester}, then proved in \cite{RiderVirag} for Ginibre matrices, and then in more generality in \cite{AHM2011, LebleSerfaty}.
\end{itemize}
\end{remark}

\begin{remark}
As a sanity check, note that $c_{1}$ and $c_{2}$ in \eqref{coeff for exp and var} satisfy $c_{1}<0$ and $c_{2}>0$ for all $\mathfrak{s}\in \mathbb{R}$, which is consistent with $\kappa_{1} = \mathbb{E}[N(D_{r})]\leq n$ and $\kappa_{2} = \mathrm{Var}[N(D_{r})]  >0$. Note also that $c_{2}$ in \eqref{coeff for exp and var} decays exponentially fast as $\mathfrak{s} \to + \infty$, which suggests that for any $\epsilon>0$, $\mathrm{Var}[N(D_{r})]$ decays very fast as $n \to + \infty$ and simultaneously $n^{-\epsilon}\mathfrak{s} \to + \infty$ (but this is only a heuristic since the error terms in \eqref{formula cumulant edge} are proved for fixed $\mathfrak{s}$). The error terms in \eqref{formulas cumulants outside the bulk} and \eqref{cov asymp} are far from optimal and could be easily improved if needed, but we do not pursue that here (see also \cite[Theorem 1.7]{Rider} where an exponentially small error term was obtained for $\mathrm{Cov}[N(D_{r_{1}}),N(D_{r_{2}})]$ in the Ginibre case). 
\end{remark}
\begin{proof}[Proof of Corollary \ref{coro:counting function}]
Proof of parts (a), (b), (c) and (d): the asymptotics \eqref{formula cumulant bulk}, \eqref{formula cumulant edge}, \eqref{formulas cumulants outside the bulk} and \eqref{cov asymp} are obtained by combining \eqref{cumulant}--\eqref{joint cumulant} with Theorem \ref{thm:main thm} (using in particular \eqref{der of main result}). To obtain \eqref{formula cumulant bulk}, one needs to further note that 
\begin{align*}
& \partial_{u}^{j} \Big( \mathcal{F}(t,e^{u}) + \mathcal{F}(t,e^{-u}) \Big) = 0 \quad \mbox{for } j \mbox{ odd}, \qquad \partial_{u}^{j} \Big( \mathcal{F}(t,e^{u}) - \mathcal{F}(t,e^{-u}) \Big) = 0 \quad \mbox{for } j \mbox{ even}, \\
& t \mapsto \mathcal{G}_{j}(t):=\partial_{u}^{j} \mathcal{G}(t,e^{u})\big|_{u=0} \quad \mbox{satisfies } \mathcal{G}_{j}(t)=\mathcal{G}_{j}(-t) \mbox{ for $j$ odd and } \mathcal{G}_{j}(t)=-\mathcal{G}_{j}(-t) \mbox{ for $j$ even},
\end{align*}
from which it easily follows that the coefficients $c_{j}, d_{j}, e_{j}$ of \eqref{cjdjej bulk} satisfy $c_{j}=e_{j}=0$ for $j$ odd and $d_{j}=0$ for $j$ even. The simplified formulas \eqref{mean bulk}, \eqref{var bulk}, \eqref{coeff for exp and var} are then obtained using integration by parts. We now turn to the proof of part (e). For this, recall that there exists $\delta>0$ such that \eqref{asymp in main thm} holds uniformly for $u_{1},\ldots,u_{m+1} \in \{z \in \mathbb{C}: |z| \leq \delta\}$. Hence, using Theorem \ref{thm:main thm} with $p=m+1$ and
\begin{align*}
& u_{j} = \pi^{1/4}\frac{t_{j}}{\sqrt{br_{\smash{j}}^{b}} \; n^{1/4}}, \qquad j=1,\ldots,m, & & u_{m+1}:= \frac{t_{m+1}}{\sqrt{c_{2}}\; n^{1/4}},
\end{align*}
where $t_{1},\ldots,t_{m+1}\in \mathbb{R}$ are arbitrary but fixed and $c_{2}$ is as in \eqref{coeff for exp and var}, we obtain
\begin{align*}
\mathbb{E}\bigg[ \prod_{j=1}^{m+1}e^{t_{j}N_{j}} \bigg] = \exp \bigg( \sum_{j=1}^{m+1}\frac{t_{j}^{2}}{2} + \bigO(n^{-\frac{1}{4}}) \bigg), \qquad \mbox{as } n \to + \infty,
\end{align*}
which implies the claim.
\end{proof}
\subsection*{Outline of the proof of Theorem \ref{thm:main thm}.}
Since $w$ is rotation-invariant, $D_{n}$ can be identically expressed in terms of one-fold integrals (albeit not being a Selberg integral). This fact is well-known and has already been used in different contexts, see e.g. \cite{Rider, FyoSommer, DeanoSimm, EZ2018}. For convenience, we also give a proof of this result here.
\begin{lemma}\label{lemma: exact identity}
Let $\mathsf{w}$ be a rotation invariant weight satisfying
\begin{align*}
\int_{0}^{+\infty} u^{j}\mathsf{w}(u)du < + \infty, \qquad \mbox{for all } j \geq 0.
\end{align*}
Then
\begin{align*}
\frac{1}{n!} \int_{\mathbb{C}}\ldots \int_{\mathbb{C}} \prod_{1 \leq j < k \leq n} |z_{k} -z_{j}|^{2} \prod_{j=1}^{n}\mathsf{w}(z_{j}) d^{2}z_{j} = (2\pi)^{n} \prod_{j=0}^{n-1} \int_{0}^{+\infty} u^{2j+1}\mathsf{w}(u)du.
\end{align*}
\end{lemma}
\begin{proof}
It follows from e.g. \cite[Lemma 2.1]{WebbWong} that 
\begin{align}\label{lol7}
\frac{1}{n!} \int_{\mathbb{C}}\ldots \int_{\mathbb{C}} \prod_{1 \leq j < k \leq n} |z_{k} -z_{j}|^{2} \prod_{j=1}^{n}\mathsf{w}(z_{j}) d^{2}z_{j} = \det \left( \int_{\mathbb{C}} z^{j} \overline{z}^{k} \mathsf{w}(z) d^{2}z \right)_{j,k=0}^{n-1}.
\end{align}
Since $\mathsf{w}$ is rotation-invariant, 
\begin{align*}
\int_{\mathbb{C}} z^{j} \overline{z}^{k} \mathsf{w}(z) d^{2}z = \begin{cases}
0, & \mbox{if } j \neq k, \\
2\pi \int_{0}^{+\infty}u^{2j+1}\mathsf{w}(u)du, & \mbox{if } j=k,
\end{cases}
\end{align*}
and the claim follows. 
\end{proof}
\begin{remark}\label{remark:asymp of Zn}
Recall that $Z_{n}$ is the normalization constant of \eqref{def of point process}. Applying Lemma \ref{lemma: exact identity} to $\mathsf{w}(z)=|z|^{2\alpha} e^{-n |z|^{2b}}$, we obtain
\begin{align}\label{explicit formula for Zn}
Z_{n} := \frac{1}{n!} \int_{\mathbb{C}}\ldots \int_{\mathbb{C}} \prod_{1 \leq j < k \leq n} |z_{k} -z_{j}|^{2} \prod_{j=1}^{n}|z_{j}|^{2\alpha} e^{-n |z_{j}|^{2b}} d^{2}z_{j} = n^{-\frac{n^{2}}{2b}}n^{-\frac{1+2\alpha}{2b}n} \frac{\pi^{n}}{b^{n}} \prod_{j=1}^{n} \Gamma(\tfrac{j+\alpha}{b}). 
\end{align}
The above right-hand side can be easily expanded as $n \to + \infty$ using \cite[formula 5.11.1]{NIST}  
\begin{align*}
\log \Gamma(z) = z \log z - z - \frac{\log z}{2} + \frac{\log 2\pi}{2} + \frac{1}{12 z} + \bigO(z^{-3}), \qquad \mbox{as } z \to + \infty,
\end{align*}
and we obtain
\begin{align}
\log Z_{n} & = - \frac{3+2 \log b}{4b}n^{2} - \frac{1}{2} n \log n + \bigg( \frac{\log (2\pi)}{2} + \frac{b-2\alpha-1}{2b}(1+\log b) + \log \frac{\pi}{b} \bigg)n \nonumber \\
& + \frac{1-3b+b^{2}+6\alpha-6b \alpha + 6\alpha^{2}}{12b}\log n + \mathfrak{g}(b,\alpha) + \bigO(n^{-2}), \qquad \mbox{as } n \to +\infty, \label{asymp of Zn}
\end{align}
for a certain constant $\mathfrak{g}(b,\alpha)$. As noticed in \cite[Proposition 1.4]{CLM2021}, if $b = \frac{n_{1}}{n_{2}}$ for some $n_{1},n_{2}\in \mathbb{N}_{>0}$, then $\mathfrak{g}(b,\alpha)$ is explicitly given by\footnote{$\mathfrak{g}(b,\alpha)$ here corresponds to $d(\frac{1}{b},\frac{\alpha}{b}-1)$ in \cite{CLM2021}.}
\begin{align*}
\mathfrak{g}(b,\alpha) & = n_{1}n_{2}\zeta'(-1) + \frac{b(n_{2}-n_{1})+2n_{1}\alpha}{4b}\log(2\pi) \\
& - \frac{1-3b+b^{2}+6\alpha-6b\alpha+6\alpha^{2}}{12b}\log n_{1} - \sum_{j=1}^{n_{2}}\sum_{k=1}^{n_{1}} \log G \bigg( \frac{j+\frac{\alpha}{b}-1}{n_{2}} + \frac{k}{n_{1}} \bigg),
\end{align*}
where $G$ is Barnes' $G$-function. Thus one can obtain the large $n$ asymptotics of $D_{n}$ by combining \eqref{def of Dn as n fold integral}, \eqref{asymp in main thm} and \eqref{asymp of Zn}.
\end{remark}
For convenience, let us write $\omega$ (which was defined in \eqref{def of omega}) as
\begin{align}\label{def of omegaell}
\omega (x) = \sum_{\ell=1}^{p+1}\omega_{\ell} \mathbf{1}_{[0,r_{\ell})}(x), \qquad \omega_{\ell} := \begin{cases}
e^{u_{\ell}+\ldots+u_{p}}-e^{u_{\ell+1}+\ldots+u_{p}}, & \mbox{if } \ell < p, \\
e^{u_{p}}-1, & \mbox{if } \ell=p, \\
1, & \mbox{if } \ell=p+1,
\end{cases}
\end{align}
where $r_{p+1}:=+\infty$. Applying Lemma \ref{lemma: exact identity} to $\mathsf{w}=w$, we immediately get the following exact identity for $D_{n}$:
\begin{align}
D_{n} & = (2\pi)^{n} \prod_{j=0}^{n-1} \int_{0}^{+\infty} u^{2j+1+2\alpha}e^{-n u^{2b}}\omega(u)du \nonumber \\
& = n^{-\frac{n^{2}}{2b}}n^{-\frac{1+2\alpha}{2b}n} \frac{\pi^{n}}{b^{n}} \prod_{j=1}^{n} \bigg(\sum_{\ell=1}^{p} \omega_{\ell} \gamma(\tfrac{j+\alpha}{b},nr_{\ell}^{2b}) + \Gamma(\tfrac{j+\alpha}{b}) \bigg), \label{main exact formula}
\end{align}
where $\gamma(a,z)$ is the incomplete gamma function
\begin{align*}
\gamma(a,z) = \int_{0}^{z}t^{a-1}e^{-t}dt.
\end{align*}
As can be seen from \eqref{main exact formula}, to obtain the large $n$ asymptotics of $D_{n}$, we need the asymptotics of $\gamma(a,z)$ as $z \to +\infty$ uniformly for $a\in [\frac{1+\alpha}{b},\frac{z}{b r_{1}^{2b}}+\frac{\alpha}{b}]$. These asymptotics are already known and are stated in the following lemmas. 
\begin{lemma}\label{lemma:various regime of gamma}(Taken from \cite[formula 8.11.2]{NIST}).
Let $a>0$ be fixed. As $z \to +\infty$,
\begin{align*}
\gamma(a,z) = \Gamma(a) + \bigO(e^{-\frac{z}{2}}).
\end{align*}
\end{lemma}
\begin{lemma}\label{lemma: uniform}(Taken from \cite[Section 11.2.4]{Temme}).
We have
\begin{align*}
& \frac{\gamma(a,z)}{\Gamma(a)} = \frac{1}{2}\mathrm{erfc}(-\eta \sqrt{a/2}) - R_{a}(\eta), \qquad R_{a}(\eta) = \frac{e^{-\frac{1}{2}a \eta^{2}}}{2\pi i}\int_{-\infty}^{\infty}e^{-\frac{1}{2}a u^{2}}g(u)du,
\end{align*}
where $\mathrm{erfc}$ is defined in \eqref{def of erfc}, $\lambda := \frac{z}{a}$, $g(u) := \frac{dt}{du}\frac{1}{\lambda-t}+\frac{1}{u+i\eta}$,
\begin{align}\label{lol8}
& \eta = (\lambda-1)\sqrt{\frac{2 (\lambda-1-\log \lambda)}{(\lambda-1)^{2}}}, \qquad  u = -i(t-1)\sqrt{\frac{2(t-1-\log t)}{(t-1)^{2}}},
\end{align}
	where the principal branch is used for the roots. In particular, $\eta>0$ for $\lambda > 1$, $\eta<0$ for $\lambda < 1$, and $u \in \mathbb{R}$ for $t \in \mathcal{L}:=\{\frac{\theta}{\sin \theta}e^{i\theta}:-\pi<\theta<\pi\}$. Furthermore, as $a \to + \infty$, uniformly for $z \in [0,\infty)$,
\begin{align}\label{asymp of Ra}
& R_{a}(\eta) \sim \frac{e^{-\frac{1}{2}a \eta^{2}}}{\sqrt{2\pi a}}\sum_{j=0}^{\infty} \frac{c_{j}(\eta)}{a^{j}},
\end{align}
where all coefficients $c_{j}(\eta)$ are bounded functions of $\eta \in \mathbb{R}$ (i.e. bounded for $\lambda \in [0,\infty)$). The first two coefficients are given by (see \cite[p. 312]{Temme})
\begin{align*}
c_{0}(\eta) = \frac{1}{\lambda-1}-\frac{1}{\eta}, \qquad c_{1}(\eta) = \frac{1}{\eta^{3}}-\frac{1}{(\lambda-1)^{3}}-\frac{1}{(\lambda-1)^{2}}-\frac{1}{12(\lambda-1)}.
\end{align*}
In particular, the following hold:
\item[(i)] Let $\delta>1$ be fixed, and let $z=\lambda a$. As $a \to +\infty$, uniformly for $\lambda \geq 1+\delta$,
\begin{align*}
\gamma(a,z) = \Gamma(a)\big(1 + \bigO(e^{-\frac{a \eta^{2}}{2}})\big).
\end{align*}
\item[(ii)] Let $z=\lambda a$. As $a \to +\infty$, uniformly for $\lambda$ in compact subsets of $(0,1)$,
\begin{align*}
\gamma(a,z) = \Gamma(a)\bigO(e^{-\frac{a \eta^{2}}{2}}).
\end{align*}
\end{lemma}

\section{Proof of Theorem \ref{thm:main thm}}\label{section:proof}
In this section, $\log$ always denotes the principal branch of the logarithm, and $c$ and $C$ denote positive constants which may change within a computation. 

\medskip Let $M'$ be a large integer independent of $n$, let $\epsilon > 0$ be a small constant independent of $n$, and let $M:=M'\sqrt{\log n}$. Define
\begin{align*}
& j_{\ell,-}:=\lceil \tfrac{bnr_{\ell}^{2b}}{1+\epsilon} - \alpha \rceil, \qquad j_{\ell,+} := \lfloor  \tfrac{bnr_{\ell}^{2b}}{1-\epsilon} - \alpha \rfloor, \qquad \ell=1,\ldots,m, \\
& j_{m+1,-}:=\lceil \tfrac{n}{1+\epsilon} - \alpha \rceil, \qquad j_{m+1,+} := n,
\end{align*}
$j_{0,-}:=1$, $j_{0,+}:=M'$. We take $\epsilon$ sufficiently small such that
\begin{align*}
\frac{br_{\ell}^{2b}}{1-\epsilon} < \frac{br_{\ell+1}^{2b}}{1+\epsilon}, \qquad \mbox{for all } \ell \in \{1,\ldots,m\}.
\end{align*}
Using \eqref{main exact formula}, we split $\log D_{n}$ into $2m+4$ parts
\begin{align}\label{log Dn as a sum of sums}
\log D_{n} = S_{-1} + S_{0} + \sum_{k=1}^{m}(S_{2k-1}+S_{2k}) + S_{2m+1} + S_{2m+2},
\end{align}
with 
\begin{align}
& S_{-1} = -\tfrac{1}{2b}n^{2}\log n - \tfrac{1+2\alpha}{2b}n \log n + n \log \tfrac{\pi}{b}, & & S_{0} = \sum_{j=1}^{M'} \log \bigg( \sum_{\ell=1}^{p+1} \omega_{\ell} \gamma(\tfrac{j+\alpha}{b},nr_{\ell}^{2b}) \bigg), \label{def of Sm1 and S0} \\
& S_{2k-1} = \sum_{j=j_{k-1,+}+1}^{j_{k,-}-1} \hspace{-0.3cm} \log \bigg( \sum_{\ell=1}^{p+1} \omega_{\ell} \gamma(\tfrac{j+\alpha}{b},nr_{\ell}^{2b}) \bigg), & & S_{2k} = \sum_{j=j_{k,-}}^{j_{k,+}} \log \bigg( \sum_{\ell=1}^{p+1} \omega_{\ell} \gamma(\tfrac{j+\alpha}{b},nr_{\ell}^{2b}) \bigg), \label{def of S bulk} \\
& S_{2m+1}=\sum_{j=j_{m,+}+1}^{j_{m+1,-}-1} \hspace{-0.3cm} \log \bigg( \sum_{\ell=1}^{p+1} \omega_{\ell} \gamma(\tfrac{j+\alpha}{b},nr_{\ell}^{2b}) \bigg), & & S_{2m+2} = \sum_{j=j_{m+1,-}}^{n} \log \bigg( \sum_{\ell=1}^{p+1} \omega_{\ell} \gamma(\tfrac{j+\alpha}{b},nr_{\ell}^{2b}) \bigg). \label{def of S edge}
\end{align}
For convenience, we also define
\begin{align}\label{def of Omega j}
\Omega_{\ell} = \sum_{j=\ell}^{p+1}\omega_{j} = \begin{cases}
e^{u_{\ell}+\ldots+u_{p}}, & \mbox{if } \ell \leq p, \\
1 & \mbox{if } \ell=p+1,
\end{cases}
\end{align}
so that $\omega$ can be rewritten as
\begin{align*}
\omega (x) = \sum_{\ell=1}^{p+1}\omega_{\ell} \mathbf{1}_{[0,r_{\ell})}(x) = \sum_{\ell=1}^{p+1}\Omega_{\ell} \mathbf{1}_{[r_{\ell-1},r_{\ell})}(x).
\end{align*}
\begin{lemma}\label{lemma: S0}
For any $x_{1},\ldots,x_{p} \in \mathbb{R}$, there exists $\delta > 0$ such that
\begin{align}\label{asymp of S0}
S_{0} = M' \log \Omega_{1} + \sum_{j=1}^{M'} \log \Gamma(\tfrac{j+\alpha}{b}) + \bigO(e^{-cn}), \qquad \mbox{as } n \to + \infty,
\end{align}
uniformly for $u_{1} \in \{z \in \mathbb{C}: |z-x_{1}|\leq \delta\},\ldots,u_{p} \in \{z \in \mathbb{C}: |z-x_{p}|\leq \delta\}$.
\end{lemma}
\begin{proof}
By \eqref{def of Sm1 and S0} and Lemma \ref{lemma:various regime of gamma}, as $n \to +\infty$ we have
\begin{align*}
S_{0} & = \sum_{j=1}^{M'} \log \bigg( \sum_{\ell=1}^{p+1} \omega_{\ell} \big[\Gamma(\tfrac{j+\alpha}{b}) + \bigO(e^{-cn}) \big] \bigg) = \sum_{j=1}^{M'} \log \big( \Omega_{1}\Gamma(\tfrac{j+\alpha}{b}) \big) + \bigO(e^{-cn}),
\end{align*}
where the first error term in the above expression is independent of $u_{1},\ldots,u_{p}$. This clearly implies the claim.
\end{proof}

\begin{lemma}\label{lemma: S2km1}
Let $k \in \{1,\ldots,m+1\}$. For any $x_{1},\ldots,x_{p} \in \mathbb{R}$, there exists $\delta > 0$ such that 
\begin{align*}
S_{2k-1} = 
(j_{k,-}-j_{k-1,+}-1) \log \Omega_{k} + \sum_{j=j_{k-1,+}+1}^{j_{k,-}-1}  \hspace{-0.3cm} \log \Gamma(\tfrac{j+\alpha}{b}) + \bigO(e^{-cn})
\end{align*}
as $n \to +\infty$ uniformly for $u_{1} \in \{z \in \mathbb{C}: |z-x_{1}|\leq \delta\},\ldots,u_{p} \in \{z \in \mathbb{C}: |z-x_{p}|\leq \delta\}$.
\end{lemma}

\begin{proof}
Recall that $S_{2k-1}$ is defined in \eqref{def of S bulk} and \eqref{def of S edge}, and define $a_{j} := \tfrac{j+\alpha}{b}$, $\lambda_{j,\ell} := \frac{bnr_{\ell}^{2b}}{j+\alpha}$ and 
\begin{align}\label{def etajl}
\eta_{j,\ell} := (\lambda_{j,\ell}-1)\sqrt{\frac{2 (\lambda_{j,\ell}-1-\ln \lambda_{j,\ell})}{(\lambda_{j,\ell}-1)^{2}}}.
\end{align}
Let us treat the case $k \geq 2$ first. We use Lemma \ref{lemma: uniform} (i)--(ii) with $a$ and $\lambda$ replaced by $a_{j}$ and $\lambda_{j,\ell}$ respectively, where $j \in \{j_{k-1,+}+1,\ldots,j_{k,-}-1\}$ and $\ell\in \{1,\ldots,p\}$. This gives
\begin{align}\label{lol9}
S_{2k-1} & = \hspace{-0.15cm} \sum_{j=j_{k-1,+}+1}^{j_{k,-}-1}  \hspace{-0.3cm} \log \Gamma(\tfrac{j+\alpha}{b}) + \hspace{-0.15cm} \sum_{j=j_{k-1,+}+1}^{j_{k,-}-1} \hspace{-0.3cm}
\log \bigg( \sum_{\ell=1}^{k-1} \omega_{\ell}  \bigO(e^{-\frac{a_{j}\eta_{j,\ell}^{2}}{2}}) + \sum_{\ell=k}^{p} \omega_{\ell}  \big(1+\bigO(e^{-\frac{a_{j}\eta_{j,\ell}^{2}}{2}})\big) + 1 \bigg), 
\end{align}
as $n \to + \infty$, where the above error terms are independent of $\omega_{1},\ldots,\omega_{p}$. Note that for each $k \in \{2,\ldots,m+1\}$, there exist positive constants $\{c_{j},c_{j}'\}_{j=1}^{3}$ such that $c_{1} n \leq  a_{j} \leq c_{1}'n$, $c_{2} \leq |\lambda_{j,\ell}-1| \leq c_{2}'$ and $c_{3} \leq \eta_{j,\ell}^{2} \leq c_{3}'$ hold for all $n$ sufficiently large, for all $j \in \{j_{k-1,+}+1,\ldots,j_{k,-}-1\}$ and for all $\ell\in \{1,\ldots,p\}$. This proves the claim for $k=2,\ldots,m+1$ with $c=\frac{c_{1}c_{3}}{2}$. The proof for $k=1$ is only slightly different. By Lemma \ref{lemma: uniform} (i), for any $\epsilon'>0$ there exist $A=A(\epsilon'),C=C(\epsilon')>0$ such that $|\frac{\gamma(a,z)}{\Gamma(a)}-1| \leq Ce^{-\frac{a\eta^{2}}{2}}$ for all $a \geq A$, for all $\lambda=\frac{z}{a} \in [1+\epsilon',+\infty]$, and where $\eta$ is defined by \eqref{lol8}. Let us take $\epsilon'=\frac{\epsilon}{2}$ and choose $M'$ large enough so that $a_{j} = \frac{j+\alpha}{b} \geq A(\frac{\epsilon}{2})$ for all $j \in \{M'+1,\ldots,j_{1,-}-1\}$. Hence, as in \eqref{lol9} we find
\begin{align*}
S_{1} & = \sum_{j=M'+1}^{j_{1,-}-1}  \hspace{-0.3cm} \log \Gamma(\tfrac{j+\alpha}{b}) + \sum_{j=M'+1}^{j_{1,-}-1} \hspace{-0.3cm}
\log \bigg( \sum_{\ell=1}^{p} \omega_{\ell}  \big(1+\bigO(e^{-\frac{a_{j}\eta_{j,\ell}^{2}}{2}})\big) + 1 \bigg), \qquad \mbox{as } n \to + \infty.
\end{align*}
It is easy to check that for each $\ell \in \{1,\ldots,p\}$, the quantity $a_{j}\eta_{j,\ell}^{2}$ decreases as $j$ increases from $M'+1$ to $j_{1,-}-1$. Hence
\begin{align*}
\frac{a_{j}\eta_{j,\ell}^{2}}{2} \geq \frac{a_{j_{1,-}-1}\eta_{j_{1,-}-1,\ell}^{2}}{2} \geq cn, \qquad \mbox{for all } j \in \{M'+1,\ldots,j_{1,-}-1\}, \; \ell \in \{1,\ldots,p\},
\end{align*}
for a sufficiently small $c>0$. This proves the claim for $k=1$.
\end{proof}
We now turn our attention to the sums $S_{2k}$, $k=1,\ldots,m+1$.
\begin{lemma}
For any $x_{1},\ldots,x_{p} \in \mathbb{R}$, there exists $\delta > 0$ such that
\begin{align*}
& S_{2k}=S_{2k}^{(1)}+S_{2k}^{(2)}+S_{2k}^{(3)} + \bigO(e^{-cn}), & & k =1,2,\ldots,m,
\end{align*}
as $n \to \infty$ uniformly for $u_{1} \in \{z \in \mathbb{C}: |z-x_{1}|\leq \delta\},\ldots,u_{p} \in \{z \in \mathbb{C}: |z-x_{p}|\leq \delta\}$, where
\begin{align}\label{asymp prelim of S2kpvp}
& S_{2k}^{(v)} = \sum_{j:\lambda_{j,k}\in I_{v}}  \log \bigg( \omega_{k} \gamma(a_{j},z_{k}) +   \Omega_{k+1} \Gamma(\tfrac{j+\alpha}{b}) \bigg), \quad v=1,2,3, \quad k=1,\ldots,m
\end{align}
with 
\begin{align*}
a_{j}:=\frac{j+\alpha}{b}, \qquad z_{k}:=nr_{k}^{2b}, \qquad \lambda_{j,k} := \frac{z_{k}}{a_{j}} = \frac{bnr_{k}^{2b}}{j+\alpha},
\end{align*}
and where
\begin{align*}
I_{1} = [1-\epsilon,1-\tfrac{M}{\sqrt{n}}), \qquad I_{2} = [1-\tfrac{M}{\sqrt{n}},1+\tfrac{M}{\sqrt{n}}], \qquad I_{3} = (1+\tfrac{M}{\sqrt{n}},1+\epsilon].
\end{align*}
Similarly, for any $x_{1},\ldots,x_{p} \in \mathbb{R}$, there exists $\delta > 0$ such that
\begin{align*}
S_{2m+2}=S_{2m+2}^{(2)}+S_{2m+2}^{(3)} + \bigO(e^{-cn}),
\end{align*}
as $n \to \infty$ uniformly for $u_{1} \in \{z \in \mathbb{C}: |z-x_{1}|\leq \delta\},\ldots,u_{p} \in \{z \in \mathbb{C}: |z-x_{p}|\leq \delta\}$, where
\begin{align}\label{asymp prelim of S2mp2pvp}
& S_{2m+2}^{(v)} = \sum_{j:\lambda_{j,m+1}\in I_{v}}  \log \bigg( \omega_{m+1} \gamma(a_{j},z_{m+1}) +   \Omega_{m+2} \Gamma(\tfrac{j+\alpha}{b}) \bigg),  \quad v=2',3,
\end{align}
with 
\begin{align*}
a_{j}:=\frac{j+\alpha}{b}, \qquad z_{m+1}:=nr_{m+1}^{2b}, \qquad \lambda_{j,m+1} := \frac{z_{m+1}}{a_{j}} = \frac{bnr_{m+1}^{2b}}{j+\alpha},
\end{align*}
and where
\begin{align*}
I_{2'} = [\tfrac{1+\frac{\mathcal{R}}{\sqrt{n}}}{1+\frac{\alpha}{n}},1+\tfrac{M}{\sqrt{n}}], \qquad I_{3} = (1+\tfrac{M}{\sqrt{n}},1+\epsilon], \qquad \mathcal{R} := \sqrt{2b} \, \mathfrak{s}.
\end{align*}
\end{lemma}
\begin{proof}
By \eqref{def of S bulk}, \eqref{def of S edge} and Lemma \ref{lemma: uniform} (i)--(ii), we have
\begin{align}\label{lol13}
S_{2k} & = \sum_{j=j_{k,-}}^{j_{k,+}} \log  \Gamma(\tfrac{j+\alpha}{b}) + \sum_{j=j_{k,-}}^{j_{k,+}} \log \bigg( \sum_{\ell=1}^{k-1} \omega_{\ell} \bigO(e^{-cn}) +  \omega_{k} \frac{\gamma(\tfrac{j+\alpha}{b},nr_{k}^{2b})}{\Gamma(\tfrac{j+\alpha}{b})} + \sum_{\ell=k+1}^{p+1} \omega_{\ell}\big(1+\bigO(e^{-cn})\big) \bigg) ,
\end{align}
as $n \to + \infty$, where the above error terms are independent of $\omega_{1},\ldots,\omega_{p}$. Let us choose $\delta>0$ sufficiently small such that
\begin{align*}
\omega_{k} \frac{\gamma(\tfrac{j+\alpha}{b},nr_{k}^{2b})}{\Gamma(\tfrac{j+\alpha}{b})} + \Omega_{k+1}
\end{align*}
remains bounded away from the interval $(-\infty,0]$ as $n \to + \infty$ uniformly for $u_{1} \in \{z \in \mathbb{C}: |z-x_{1}|\leq \delta\},\ldots,u_{p} \in \{z \in \mathbb{C}: |z-x_{p}|\leq \delta\}$ and for $j \in \{j_{k,-},\ldots,j_{k,+}\}$. By \eqref{lol13},
\begin{align*}
S_{2k} & = \sum_{j=j_{k,-}}^{j_{k,+}} \log  \Gamma(\tfrac{j+\alpha}{b}) + \sum_{j=j_{k,-}}^{j_{k,+}} \log \bigg( \Omega_{k+1} +  \omega_{k} \frac{\gamma(\tfrac{j+\alpha}{b},nr_{k}^{2b})}{\Gamma(\tfrac{j+\alpha}{b})} \bigg) + \bigO(e^{-cn}),
\end{align*}
as $n \to + \infty$ uniformly for $u_{1} \in \{z \in \mathbb{C}: |z-x_{1}|\leq \delta\},\ldots,u_{p} \in \{z \in \mathbb{C}: |z-x_{p}|\leq \delta\}$. Note that $\lambda_{j,k} \in [1-\epsilon,1+\epsilon]$ for each $k \in \{1,\ldots,m\}$ and $j \in \{j_{k,-},\ldots,j_{k,+}\}$, while $\lambda_{j,m+1} \in [\frac{1+\frac{\mathcal{R}}{\sqrt{n}}}{1+\frac{\alpha}{n}},1+\epsilon]$ for each $j \in \{j_{m+1,-},\ldots,j_{m+1,+}\}$. The claim now follows directly by splitting $S_{2k}$, $k=1,\ldots,m$ into three parts and $S_{2m+2}$ into two parts.
\end{proof}
For $k=1,\ldots,m$, define $g_{k,-} := \lceil \frac{bnr_{k}^{2b}}{1+\frac{M}{\sqrt{n}}}-\alpha \rceil$, $g_{k,+} := \lfloor \frac{bnr_{k}^{2b}}{1-\frac{M}{\sqrt{n}}}-\alpha \rfloor$, and  
\begin{align*}
& \theta_{k,-}^{(n,M)} := g_{k,-} - \bigg( \frac{bn r_{k}^{2b}}{1+\frac{M}{\sqrt{n}}} - \alpha \bigg) = \bigg\lceil \frac{bn r_{k}^{2b}}{1+\frac{M}{\sqrt{n}}} - \alpha \bigg\rceil - \bigg( \frac{bn r_{k}^{2b}}{1+\frac{M}{\sqrt{n}}} - \alpha \bigg), \\
& \theta_{k,+}^{(n,M)} := \bigg( \frac{bn r_{k}^{2b}}{1-\frac{M}{\sqrt{n}}} - \alpha \bigg) - g_{k,+} = \bigg( \frac{bn r_{k}^{2b}}{1-\frac{M}{\sqrt{n}}} - \alpha \bigg) - \bigg\lfloor \frac{bn r_{k}^{2b}}{1-\frac{M}{\sqrt{n}}} - \alpha \bigg\rfloor.
\end{align*}
Clearly, $\theta_{k,-}^{(n,M)},\theta_{k,+}^{(n,M)} \in [0,1)$, and formally we can write
\begin{align}\label{sums lambda j}
& \sum_{j:\lambda_{j,k}\in I_{3}} = \sum_{j=j_{k,-}}^{g_{k,-}-1}, \qquad \sum_{j:\lambda_{j,k}\in I_{2}} = \sum_{j= g_{k,-}}^{g_{k,+}}, \qquad \sum_{j:\lambda_{j,k}\in I_{1}} = \sum_{j= g_{k,+}+1}^{j_{k,+}}. 
\end{align}
Define also $g_{m+1,-} := \lceil \frac{bnr_{m+1}^{2b}}{1+\frac{M}{\sqrt{n}}}-\alpha \rceil$ and
\begin{align*}
& \theta_{m+1,-}^{(n,M)} := g_{m+1,-} - \bigg( \frac{bnr_{m+1}^{2b}}{1+\frac{M}{\sqrt{n}}} - \alpha \bigg) = \bigg\lceil \frac{bnr_{m+1}^{2b}}{1+\frac{M}{\sqrt{n}}} - \alpha \bigg\rceil - \bigg( \frac{bn r_{m+1}^{2b}}{1+\frac{M}{\sqrt{n}}} - \alpha \bigg) \in [0,1).
\end{align*}
Formally, we have
\begin{align}\label{sums lambda j 2}
& \sum_{j:\lambda_{j,m+1}\in I_{3}} = \sum_{j=j_{m+1,-}}^{g_{m+1,-}-1}, \qquad \sum_{j:\lambda_{j,m+1}\in I_{2'}} = \sum_{j= g_{m+1,-}}^{n}. 
\end{align}
We collect in the following lemma some straightforward but useful asymptotic formulas.
\begin{lemma}
As $n \to + \infty$, we have
\begin{align}
\sum_{j=g_{k,+}+1}^{j_{k,+}} 1 & = j_{k,+}-g_{k,+} = j_{k,+}- \bigg( \frac{bn r_{k}^{2b}}{1-\frac{M}{\sqrt{n}}} - \alpha \bigg) + \theta_{k,+}^{(n,M)} \nonumber \\
& = j_{k,+}-br_{k}^{2b} n - bM r_{k}^{2b} \sqrt{n} - bM^{2}r_{k}^{2b}+\alpha+\theta_{k,+}^{(n,M)} - bM^{3}r_{k}^{2b}n^{-\frac{1}{2}} + \bigO(M^{4}n^{-1}), \label{sum jmax} \\
\sum_{j = j_{k,-}}^{g_{k,-}-1} 1 & = g_{k,-}-j_{k,-} = \bigg( \frac{bn r_{k}^{2b}}{1+\frac{M}{\sqrt{n}}} - \alpha \bigg)+\theta_{k,-}^{(n,M)} -j_{k,-} \nonumber \\
& = br_{k}^{2b}n - j_{k,-} - bMr_{k}^{2b}\sqrt{n} + bM^{2}r_{k}^{2b}-\alpha+\theta_{k,-}^{(n,M)} - bM^{3}r_{k}^{2b}n^{-\frac{1}{2}} + \bigO(M^{4}n^{-1}), \label{sum jmin} 
\end{align}
where \eqref{sum jmax} holds for all $k \in \{1,\ldots,m\}$ and \eqref{sum jmin} holds for all $k \in \{1,\ldots,m+1\}$. For $k=m+1$, \eqref{sum jmin} can be further expanded using $r_{m+1}=b^{-\frac{1}{2b}}(1+\frac{\mathcal{R}}{\sqrt{n}})^{\frac{1}{2b}}$, and this gives
\begin{align}
& \sum_{j = j_{m+1,-}}^{g_{m+1,-}-1} 1 = g_{m+1,-}-j_{m+1,-} = \bigg( \frac{bn r_{m+1}^{2b}}{1+\frac{M}{\sqrt{n}}} - \alpha \bigg)+\theta_{m+1,-}^{(n,M)} -j_{m+1,-} = n - j_{m+1,-} \nonumber \\
&  + (\mathcal{R} - M)\sqrt{n} + M(M-\mathcal{R})-\alpha+\theta_{m+1,-}^{(n,M)} - M^{2}(M-\mathcal{R})n^{-\frac{1}{2}} + \bigO(M^{4}n^{-1}) \label{sum jmin mp1}
\end{align}
as $n \to + \infty$.
\end{lemma}
\begin{proof}
The proof is a short computation.
\end{proof}
Our next task is to evaluate $\{S_{2k}^{(v)}\}_{k=1,\ldots,m}^{v=1,2,3}$ and $\{S_{2m+2}^{(v)}\}^{v=2,3}$. These sums are rather delicate to analyze and involve the asymptotics of $\gamma(a,z)$ in the regime $a \to + \infty$, $z \to +\infty$ with $\lambda=\frac{z}{a} \in [1-\epsilon,1+\epsilon]$.

\begin{lemma}\label{lemma:S2kp1p}
For any $k\in \{1,\ldots,m\}$ and any $x_{1},\ldots,x_{p} \in \mathbb{R}$, there exists $\delta > 0$ such that
\begin{align*}
S_{2k}^{(1)} = & \; \sum_{j= g_{k,+}+1}^{j_{k,+}} \log \Gamma(\tfrac{j+\alpha}{b}) + \Big(j_{k,+}-br_{k}^{2b} n - bM r_{k}^{2b} \sqrt{n} - bM^{2}r_{k}^{2b} \\
&  +\alpha+\theta_{k,+}^{(n,M)} - bM^{3}r_{k}^{2b}n^{-\frac{1}{2}}\Big)  \log  \Omega_{k+1} + \bigO(M^{4}n^{-1}),
\end{align*}
as $n \to +\infty$ uniformly for $u_{1} \in \{z \in \mathbb{C}: |z-x_{1}|\leq \delta\},\ldots,u_{p} \in \{z \in \mathbb{C}: |z-x_{p}|\leq \delta\}$. 
\end{lemma}
\begin{proof}
Since $I_{1}=[1-\epsilon,1-\frac{M}{\sqrt{n}})$, by \eqref{asymp prelim of S2kpvp} and Lemma \ref{lemma: uniform} we have
\begin{align*}
S_{2k}^{(1)} & = \sum_{j:\lambda_{j,k}\in I_{1}}  \log \bigg( \omega_{k} \gamma(a_{j},z_{k}) +   \Omega_{k+1} \Gamma(\tfrac{j+\alpha}{b})  \bigg) \\
& = \sum_{j:\lambda_{j,k}\in I_{1}} \bigg[ \log \Gamma(\tfrac{j+\alpha}{b}) + \log \bigg( \omega_{k} \bigg[ \frac{1}{2}\mathrm{erfc}\Big(-\eta_{j,k} \sqrt{a_{j}/2}\Big) - R_{a_{j}}(\eta_{j,k}) \bigg] + \Omega_{k+1}  \bigg)\bigg],
\end{align*}
where $\eta_{j,k} = (\lambda_{j,k}-1)\sqrt{\frac{2 (\lambda_{j,k}-1-\ln \lambda_{j,k})}{(\lambda_{j,k}-1)^{2}}}$. Since 
\begin{align}
& \eta_{j,k} = \lambda_{j,k}-1 +\bigO((\lambda_{j,k}-1)^{2}) \leq -\tfrac{M}{\sqrt{n}} + \bigO(\tfrac{M^{2}}{n}), & & \mbox{as } n \to \infty, \label{lol10} \\
& -\eta_{j,k} \sqrt{a_{j}/2} \geq  \tfrac{Mr_{k}^{b}}{\sqrt{2}} + \bigO(\tfrac{M^{2}}{\sqrt{n}}), & & \mbox{as } n \to \infty, \label{lol11}
\end{align}
uniformly for $j\in \{j: \lambda_{j,k} \in I_{1}\}$, by choosing $M'$ large enough we have
\begin{align*}
& R_{a_{j}}(\eta_{j,k}) =  \bigO(e^{-\frac{r_{k}^{2b}M^{2}}{4}}) = \bigO(n^{-10}), & & \frac{1}{2}\mathrm{erfc}\Big(-\eta_{j,k} \sqrt{a_{j}/2}\Big) = \bigO(e^{-\frac{r_{k}^{2b}M^{2}}{4}})= \bigO(n^{-10}),
\end{align*}
as $n \to + \infty$ uniformly for $j\in \{j: \lambda_{j,k} \in I_{1}\}$, and thus (using also \eqref{sums lambda j})
\begin{align}
S_{2k}^{(1)} & = \sum_{j= g_{k,+}+1}^{j_{k,+}} \bigg[ \log \Gamma(\tfrac{j+\alpha}{b}) + \log \Omega_{k+1} \bigg] + \bigO(n^{-9}) \nonumber \\
& = \sum_{j= g_{k,+}+1}^{j_{k,+}} \log \Gamma(\tfrac{j+\alpha}{b}) + (j_{k,+}-g_{k,+}) \log  \Omega_{k+1}  + \bigO(n^{-9}), \label{lol12}
\end{align}
as $n \to + \infty$. Since the error terms in \eqref{lol10} and \eqref{lol11} are clearly independent of $\omega_{1},\ldots,\omega_{p}$, the asymptotics \eqref{lol12} are uniform for $u_{1} \in \{z \in \mathbb{C}: |z-x_{1}|\leq \delta\},\ldots,u_{p} \in \{z \in \mathbb{C}: |z-x_{p}|\leq \delta\}$. The claim now follows after inserting \eqref{sum jmax} in \eqref{lol12}.
\end{proof}

\begin{lemma}\label{lemma:S2kp3p}
For any $k\in \{1,\ldots,m+1\}$ and any $x_{1},\ldots,x_{p} \in \mathbb{R}$, there exists $\delta > 0$ such that
\begin{align*}
S_{2k}^{(3)} = & \; \sum_{j=j_{k,-}}^{g_{k,-}-1} \log  \Gamma(\tfrac{j+\alpha}{b}) + \Big( br_{k}^{2b}n - j_{k,-} - bMr_{k}^{2b}\sqrt{n} + bM^{2}r_{k}^{2b} \\
& -\alpha+\theta_{k,-}^{(n,M)} - bM^{3}r_{k}^{2b}n^{-\frac{1}{2}} \Big) \log  \Omega_{k} + \bigO(M^{4}n^{-1}),
\end{align*}
as $n \to +\infty$ uniformly for $u_{1} \in \{z \in \mathbb{C}: |z-x_{1}|\leq \delta\},\ldots,u_{p} \in \{z \in \mathbb{C}: |z-x_{p}|\leq \delta\}$. For $k=m+1$, the above formula can be further expanded using $r_{m+1}=b^{-\frac{1}{2b}}(1+\frac{\mathcal{R}}{\sqrt{n}})^{\frac{1}{2b}}$, and this gives
\begin{align*}
S_{2m+2}^{(3)} = & \; \sum_{j=j_{m+1,-}}^{g_{m+1,-}-1} \log  \Gamma(\tfrac{j+\alpha}{b}) + \Big( n - j_{m+1,-} + (\mathcal{R} - M)\sqrt{n} + M(M-\mathcal{R}) \nonumber \\
&  -\alpha+\theta_{m+1,-}^{(n,M)} - M^{2}(M-\mathcal{R})n^{-\frac{1}{2}} \Big) \log \Omega_{m+1} + \bigO(M^{4}n^{-1}),
\end{align*}
as $n \to + \infty$ uniformly for $u_{1} \in \{z \in \mathbb{C}: |z-x_{1}|\leq \delta\},\ldots,u_{p} \in \{z \in \mathbb{C}: |z-x_{p}|\leq \delta\}$.
\end{lemma}
\begin{proof}
The proof is similar to Lemma \ref{lemma:S2kp1p}. Recall that $I_{3} = (1+\frac{M}{\sqrt{n}},1+\epsilon]$. Thus by \eqref{asymp prelim of S2kpvp}, \eqref{asymp prelim of S2mp2pvp} and Lemma \ref{lemma: uniform}, we have
\begin{align*}
S_{2k}^{(3)} & = \sum_{j:\lambda_{j,k}\in I_{3}} \bigg[ \log \Gamma(\tfrac{j+\alpha}{b}) + \log \bigg( \omega_{k} \bigg[ \frac{1}{2}\mathrm{erfc}\Big(-\eta_{j,k} \sqrt{a_{j}/2}\Big) - R_{a_{j}}(\eta_{j,k}) \bigg] +   \Omega_{k+1}  \bigg)\bigg],
\end{align*}
where $\eta_{j,k} = (\lambda_{j,k}-1)\sqrt{\frac{2 (\lambda_{j,k}-1-\ln \lambda_{j,k})}{(\lambda_{j,k}-1)^{2}}}$. Using the fact that 
\begin{align*}
& \eta_{j,k} = \lambda_{j,k}-1 +\bigO((\lambda_{j,k}-1)^{2}) \geq \tfrac{M}{\sqrt{n}} + \bigO(\tfrac{M^{2}}{n}), & & \mbox{as } n \to \infty, \\
& -\eta_{j,k} \sqrt{a_{j}/2} \leq - \tfrac{M r_{k}^{b}}{\sqrt{2}} + \bigO(\tfrac{M^{2}}{\sqrt{n}}), & & \mbox{as } n \to \infty,
\end{align*}
uniformly for $j\in \{j:\lambda_{j,k}\in I_{3}\}$, by choosing $M'$ large enough we have
\begin{align*}
& R_{a_{j}}(\eta_{j,k}) = \bigO(e^{-\frac{r_{k}^{2b}M^{2}}{4}}) = \bigO(n^{-10}), & & \frac{1}{2}\mathrm{erfc}\Big(-\eta_{j,k} \sqrt{a_{j}/2}\Big) = 1-\bigO(e^{-\frac{r_{k}^{2b}M^{2}}{4}}) = 1-\bigO(n^{-10}),
\end{align*}
as $n \to + \infty$ uniformly for $j\in \{j:\lambda_{j,k}\in I_{3}\}$, and thus 
\begin{align*}
S_{2k}^{(3)} & = \sum_{j=j_{k,-}}^{g_{k,-}-1} \bigg[  \log \Gamma(\tfrac{j+\alpha}{b}) + \log \Omega_{k} \bigg] +\bigO(n^{-9}) \\
& = \sum_{j=j_{k,-}}^{g_{k,-}-1} \log  \Gamma(\tfrac{j+\alpha}{b}) + (g_{k,-}-j_{k,-})  \log  \Omega_{k} +\bigO(n^{-9}),
\end{align*}
where we have also used \eqref{sums lambda j} and \eqref{sums lambda j 2} for the first equality. The claim now follows directly from \eqref{sum jmin} and \eqref{sum jmin mp1}.
\end{proof}
Our next goal is to obtain the large $n$ asymptotics of the sums $\{S_{2k}^{(2)}\}_{k=1,\ldots,m+1}$. This is the most technical part of the proof of Theorem \ref{thm:main thm}. Let us define
\begin{align*}
& M_{j,k} := \sqrt{n}(\lambda_{j,k}-1), & & \mbox{for all } k \in \{1,\ldots,m\} \mbox{ and } j \in \{j: \lambda_{j,k} \in I_{2}\}=\{g_{k,-},\ldots,g_{k,+}\}, \\
& M_{j,m+1} := \sqrt{n}(\lambda_{j,m+1}-1), & & \mbox{for all } j \in \{j: \lambda_{j,m+1} \in I_{2'}\}=\{g_{m+1,-},\ldots,n\}.
\end{align*}
By definition of $I_{2}$, as $n \to +\infty$ the points $M_{g_{k,-},k}, \ldots, M_{g_{k,+},k}$ tend to spread all over the interval $[-M,M]$ for each $k \in \{1,\ldots,m\}$. Similarly, by definition of $I_{2'}$, as $n \to +\infty$ the points $M_{g_{m+1,-},m+1}, \ldots, M_{n,m+1}$ tend to spread all over the interval $[\smash{\frac{\mathcal{R}n-\alpha \sqrt{n}}{n+\alpha}},M]$. Note also that $M_{j,k}$ decreases as $j$ increases. The following lemma will be useful to obtain the large $n$ asymptotics of $\{\smash{S_{2k}^{(2)}}\}_{k=1,\ldots,m+1}$.
\begin{lemma}\label{lemma:Riemann sum}
Let $f \in C^{3}(\mathbb{R})$ be a function such that $|f|,|f'|,|f''|,|f'''|$ are bounded. For any $k \in \{1,\ldots,m\}$, as $n \to + \infty$ we have
\begin{align}
& \sum_{j=g_{k,-}}^{g_{k,+}}f(M_{j,k}) = br_{k}^{2b} \int_{-M}^{M} f(t) dt \; \sqrt{n} - 2 b r_{k}^{2b} \int_{-M}^{M} tf(t) dt + \bigg( \frac{1}{2}-\theta_{k,-}^{(n,M)} \bigg)f(M)+ \bigg( \frac{1}{2}-\theta_{k,+}^{(n,M)} \bigg)f(-M) \nonumber \\
& + \frac{1}{\sqrt{n}}\bigg[ 3br_{k}^{2b} \int_{-M}^{M}t^{2}f(t)dt + \bigg( \frac{1}{12}+\frac{\theta_{k,-}^{(n,M)}(\theta_{k,-}^{(n,M)}-1)}{2} \bigg)\frac{f'(M)}{br_{k}^{2b}} - \bigg( \frac{1}{12}+\frac{\theta_{k,+}^{(n,M)}(\theta_{k,+}^{(n,M)}-1)}{2} \bigg)\frac{f'(-M)}{br_{k}^{2b}} \bigg] \nonumber \\
& + \bigO(M^{4}n^{-1}). \label{sum f asymp 2}
\end{align}
Also, as $n \to + \infty$, we have
\begin{align}
& \sum_{j=g_{m+1,-}}^{n}f(M_{j,m+1}) = \int_{\mathcal{R}}^{M} f(t) dt \; (\sqrt{n}+\mathcal{R}) - 2 \int_{\mathcal{R}}^{M} tf(t) dt + \bigg( \frac{1}{2}-\theta_{m+1,-}^{(n,M)} \bigg)f(M) + \bigg( \frac{1}{2} + \alpha \bigg)f(\mathcal{R}) \nonumber \\
& + \frac{1}{\sqrt{n}}\bigg[ - 2 \mathcal{R} \int_{\mathcal{R}}^{M} tf(t) dt + 3 \int_{\mathcal{R}}^{M}t^{2}f(t)dt + \bigg( \frac{1}{12}+\frac{\theta_{m+1,-}^{(n,M)}(\theta_{m+1,-}^{(n,M)}-1)}{2} \bigg)f'(M) \nonumber \\
& \hspace{1.2cm} - \frac{1+6\alpha+6\alpha^{2}}{12}f'(\mathcal{R}) \bigg] + \bigO(M^{4}n^{-1}). \label{sum f asymp 2 edge}
\end{align}
\end{lemma} 
\begin{proof}
Let $k \in \{1,\ldots,m\}$. For conciseness, we will write $M_{j}$ instead of $M_{j,k}$. Since $f \in C^{3}(\mathbb{R})$ is bounded and has bounded derivatives, we have
\begin{align}
 & \int_{M_{g_{k,+}}}^{M_{g_{k,-}}}f(t)dt = \sum_{j=g_{k,-}+1}^{g_{k,+}}\int_{M_{j}}^{M_{j-1}}f(t)dt = \sum_{j=g_{k,-}+1}^{g_{k,+}} \bigg\{ f(M_{j})(M_{j-1}-M_{j}) \nonumber \\
&  + f'(M_{j})\frac{(M_{j-1}-M_{j})^{2}}{2} + f''(M_{j})\frac{(M_{j-1}-M_{j})^{3}}{6} \bigg\} + \sum_{j=g_{k,-}+1}^{g_{k,+}}\bigO\big((M_{j-1}-M_{j})^{4}\big), \label{riemann sum 1}
\end{align}
as $n \to + \infty$.
Note that
\begin{align}
M_{j-1}-M_{j} = \frac{bn^{3/2}r_{k}^{2b}}{(j+\alpha)(j-1+\alpha)} = \frac{1}{br_{k}^{2b}n^{1/2}} + \frac{2M_{j}}{br_{k}^{2b}n}+\bigg( \frac{1}{b^{2}r_{k}^{4b}}+\frac{M_{j}^{2}}{br_{k}^{2b}} \bigg)\frac{1}{n^{3/2}}+\bigO\bigg(\frac{1+|M_{j}|}{n^{2}}\bigg), \label{diff of Mj}
\end{align}
as $n \to +\infty$ uniformly for $j \in \{g_{k,-}+1,\ldots,g_{k,+}\}$. Substituting \eqref{diff of Mj} in \eqref{riemann sum 1} and rearranging the terms, we get
\begin{align}
& \sum_{j=g_{k,-}+1}^{g_{k,+}}f(M_{j}) = br_{k}^{2b} \int_{M_{g_{k,+}}}^{M_{g_{k,-}}} f(t) dt \; \sqrt{n} -\frac{1}{\sqrt{n}} \sum_{j=g_{k,-}+1}^{g_{k,+}}\bigg( 2 M_{j}f(M_{j}) + \frac{1}{2br_{k}^{2b}}f'(M_{j}) \bigg) \nonumber \\
& -\frac{1}{n}\sum_{j=g_{k,-}+1}^{g_{k,+}}\bigg( \frac{1}{br_{k}^{2b}} f(M_{j}) + M_{j}^{2}f(M_{j}) + \frac{2}{br_{k}^{2b}} M_{j}f'(M_{j}) + \frac{1}{6b^{2}r_{k}^{4b}}f''(M_{j}) \bigg) + \bigO(M^{4}n^{-1}) \label{sum f asymp}
\end{align}
as $n \to + \infty$. In the same way as for \eqref{sum f asymp}, by replacing $f(t)$ above by $tf(t)$, $f'(t)$, $t^{2}f(t)$ and $f''(t)$, we obtain respectively
\begin{subequations}\label{lol1}
\begin{align}
\sum_{j=g_{k,-}+1}^{g_{k,+}}M_{j}f(M_{j}) & = br_{k}^{2b}\int_{M_{g_{k,+}}}^{M_{g_{k,-}}} tf(t) dt \; \sqrt{n} \nonumber \\
& \hspace{-1.2cm} - \frac{1}{\sqrt{n}}\sum_{j=g_{k,-}+1}^{g_{k,+}}\bigg( 2M_{j}^{2}f(M_{j}) + \frac{1}{2br_{k}^{2b}} f(M_{j}) + \frac{1}{2br_{k}^{2b}}M_{j}f'(M_{j}) \bigg) + \bigO(M^{4}n^{-\frac{1}{2}}),  \\
\sum_{j=g_{k,-}+1}^{g_{k,+}}f'(M_{j}) & = br_{k}^{2b} \Big( f(M_{g_{k,-}})-f(M_{g_{k,+}}) \Big) \sqrt{n} \nonumber \\
& - \frac{1}{\sqrt{n}}\sum_{j=g_{k,-}+1}^{g_{k,+}} \bigg( 2 M_{j}f'(M_{j}) + \frac{1}{2br_{k}^{2b}}f''(M_{j}) \bigg) + \bigO(M^{3}n^{-\frac{1}{2}}), \\
\sum_{j=g_{k,-}+1}^{g_{k,+}} M_{j}^{2}f(M_{j}) & = br_{k}^{2b} \int_{M_{g_{k,+}}}^{M_{g_{k,-}}} t^{2}f(t) dt \; \sqrt{n} +\bigO(M^{4}), \\
\sum_{j=g_{k,-}+1}^{g_{k,+}}f''(M_{j}) & = br_{k}^{2b} \Big( f'(M_{g_{k,-}})-f'(M_{g_{k,+}}) \Big) \sqrt{n} + \bigO(M^{2}),
\end{align}
\end{subequations}
as $n \to + \infty$. Substituting \eqref{lol1} in \eqref{sum f asymp} yields
\begin{align}
& \sum_{j=g_{k,-}}^{g_{k,+}}f(M_{j})  = br_{k}^{2b}\int_{M_{g_{k,+}}}^{M_{g_{k,-}}} f(t) dt \; \sqrt{n} - 2 b r_{k}^{2b} \int_{M_{g_{k,+}}}^{M_{g_{k,-}}} tf(t) dt + \frac{f(M_{g_{k,-}})+f(M_{g_{k,+}})}{2} \nonumber \\
& + \frac{br_{k}^{2b}}{\sqrt{n}}\bigg( 3 \int_{M_{g_{k,+}}}^{M_{g_{k,-}}} t^{2}f(t) dt + \frac{f'(M_{g_{k,-}})-f'(M_{g_{k,+}})}{12b^{2}r_{k}^{4b}} \bigg) + \bigO(M^{4}n^{-\frac{1}{2}}), \quad \mbox{as } n \to + \infty. \label{lol2}
\end{align}
Note that the sum on the left-hand side of \eqref{lol2} starts at $j=g_{k,-}$. The integrals on the right-hand side can be expanded using 
\begin{align*}
& M_{g_{k,-}} = M - \frac{\theta_{k,-}^{(n,M)}}{br_{k}^{2b}\sqrt{n}} - \frac{2M \theta_{k,-}^{(n,M)}}{br_{k}^{2b}n} + \bigO(M^{2}n^{-\frac{3}{2}}), & & \mbox{as } n \to + \infty, \\
& M_{g_{k,+}} = -M + \frac{\theta_{k,+}^{(n,M)}}{b r_{k}^{2b}\sqrt{n}} - \frac{2M \theta_{k,+}^{(n,M)}}{br_{k}^{2b}n} + \bigO(M^{2}n^{-\frac{3}{2}}), & & \mbox{as } n \to + \infty.
\end{align*}
We then find \eqref{sum f asymp 2} after a computation. We now turn to the proof of \eqref{sum f asymp 2 edge}. Again, for conciseness we will write $M_{j}$ instead of $M_{j,m+1}$. In the same way as for \eqref{lol2}, we have
\begin{align}
& \sum_{j=g_{m+1,-}}^{n}f(M_{j})  = \int_{M_{n}}^{M_{g_{m+1,-}}} f(t) dt \; (\sqrt{n} + \mathcal{R}) - 2 \int_{M_{n}}^{M_{g_{m+1,-}}} tf(t) dt + \frac{f(M_{g_{m+1,-}})+f(M_{n})}{2} \nonumber \\
& + \frac{1}{\sqrt{n}}\bigg( - 2 \mathcal{R} \int_{M_{n}}^{M_{g_{m+1,-}}} tf(t) dt + 3 \int_{M_{n}}^{M_{g_{m+1,-}}} t^{2}f(t) dt + \frac{f'(M_{g_{m+1,-}})-f'(M_{n})}{12} \bigg) + \bigO(M^{4}n^{-\frac{1}{2}}) \label{lol2 edge}
\end{align}
as $n \to + \infty$, where we have used that $br_{m+1}^{2b}=1+\frac{\mathcal{R}}{\sqrt{n}}$ and $M_{n}:=\frac{\mathcal{R}n-\alpha \sqrt{n}}{n+\alpha}$. We then obtain \eqref{sum f asymp 2 edge} from a direct computation using the expansions
\begin{align*}
& M_{g_{m+1,-}} = M - \frac{\theta_{m+1,-}^{(n,M)}}{\sqrt{n}} - \frac{(2M-\mathcal{R}) \theta_{m+1,-}^{(n,M)}}{n} + \bigO(M^{2}n^{-\frac{3}{2}}), & & \mbox{as } n \to + \infty, \\
& M_{n} = \mathcal{R} - \frac{\alpha}{\sqrt{n}} + \bigO(n^{-\frac{3}{2}}), & & \mbox{as } n \to + \infty.
\end{align*}
\end{proof}

\begin{lemma}\label{lemma:S2kp2p}
For any $k \in \{1,\ldots,m\}$ and any $x_{1},\ldots,x_{p} \in \mathbb{R}$, there exists $\delta > 0$ such that
\begin{align*}
&  S_{2k}^{(2)} = \sum_{j=g_{k,-}}^{g_{k,+}} \log \Gamma(\tfrac{j+\alpha}{b}) + \widetilde{C}_{2,k}^{(M)}\sqrt{n} + \widetilde{C}_{3,k}^{(n,M)} + \frac{1}{\sqrt{n}}\widetilde{C}_{4,k}^{(n,M)} + \bigO(M^{4}n^{-1}), \\
& \widetilde{C}_{2,k}^{(M)} = br_{k}^{2b} \int_{-M}^{M}\log \bigg( \frac{\omega_{k}}{2}\mathrm{erfc}\bigg( \frac{t \, r_{k}^{b}}{\sqrt{2}} \bigg) + \Omega_{k+1} \bigg) dt,  \\
& \widetilde{C}_{3,k}^{(n,M)} = 2 b r_{k}^{2b} \int_{-M}^{M} t \log \bigg( \frac{\omega_{k}}{2}\mathrm{erfc}\bigg( \frac{t \, r_{k}^{b}}{\sqrt{2}} \bigg) + \Omega_{k+1} \bigg) dt \\
& + \bigg( \frac{1}{2}-\theta_{k,-}^{(n,M)} \bigg)\log \bigg( \frac{\omega_{k}}{2}\mathrm{erfc}\bigg( -\frac{M \, r_{k}^{b}}{\sqrt{2}} \bigg) + \Omega_{k+1} \bigg)+ \bigg( \frac{1}{2}-\theta_{k,+}^{(n,M)} \bigg)\log \bigg( \frac{\omega_{k}}{2}\mathrm{erfc}\bigg( \frac{M \, r_{k}^{b}}{\sqrt{2}} \bigg) + \Omega_{k+1} \bigg) \\
& + b r_{k}^{2b} \int_{-M}^{M} \frac{\omega_{k}}{\Omega_{k+1}+\frac{\omega_{k}}{2}\mathrm{erfc}\big( \frac{t\, r_{k}^{b}}{\sqrt{2}} \big)}\bigg[\frac{1}{3r_{k}^{b}}-\frac{5t^{2}r_{k}^{b}}{6}\bigg]\frac{e^{-\frac{t^{2}r_{k}^{2b}}{2}}}{\sqrt{2\pi}}dt, \\
& \widetilde{C}_{4,k}^{(n,M)} = 3b r_{k}^{2b} \int_{-M}^{M}t^{2} \log \bigg( \frac{\omega_{k}}{2}\mathrm{erfc}\bigg( \frac{t \, r_{k}^{b}}{\sqrt{2}} \bigg) + \Omega_{k+1} \bigg) dt \\
& + b r_{k}^{2b} \int_{-M}^{M} \frac{\omega_{k}}{\Omega_{k+1}+\frac{\omega_{k}}{2}\mathrm{erfc}\big( \frac{t\, r_{k}^{b}}{\sqrt{2}} \big)}\frac{e^{-\frac{t^{2}r_{k}^{2b}}{2}}}{\sqrt{2\pi}} \frac{t(42-193 r_{k}^{2b} t^{2}+25 r_{k}^{4b}t^{4})}{72r_{k}^{b}}dt \\
& - \frac{b r_{k}^{2b}}{2} \int_{-M}^{M} \bigg[ \frac{\omega_{k}}{\Omega_{k+1}+\frac{\omega_{k}}{2}\mathrm{erfc}\big( \frac{t\, r_{k}^{b}}{\sqrt{2}} \big)}\frac{e^{-\frac{t^{2}r_{k}^{2b}}{2}}}{\sqrt{2\pi}} \bigg( \frac{1}{3 r_{k}^{b}}-\frac{5t^{2}r_{k}^{b}}{6} \bigg) \bigg]^{2}dt \\
& + \bigg[ \bigg( \frac{1}{2}-\theta_{k,-}^{(n,M)} \bigg)\bigg( \frac{1}{3r_{k}^{b}}-\frac{5M^{2}r_{k}^{b}}{6} \bigg) + \bigg( \frac{1}{12}+\frac{\theta_{k,-}^{(n,M)}(\theta_{k,-}^{(n,M)}-1)}{2} \bigg)\frac{1}{b r_{k}^{b}} \bigg] \frac{\omega_{k}}{\Omega_{k+1}+\frac{\omega_{k}}{2}\mathrm{erfc}\big( -\frac{M\, r_{k}^{b}}{\sqrt{2}} \big)}\frac{e^{-\frac{M^{2}r_{k}^{2b}}{2}}}{\sqrt{2\pi}} \\
& + \bigg[ \bigg( \frac{1}{2}-\theta_{k,+}^{(n,M)} \bigg)\bigg( \frac{1}{3r_{k}^{b}}-\frac{5M^{2}r_{k}^{b}}{6} \bigg) -\bigg( \frac{1}{12}+\frac{\theta_{k,+}^{(n,M)}(\theta_{k,+}^{(n,M)}-1)}{2} \bigg)\frac{1}{b r_{k}^{b}} \bigg] \frac{\omega_{k}}{\Omega_{k+1}+\frac{\omega_{k}}{2}\mathrm{erfc}\big( \frac{M\, r_{k}^{b}}{\sqrt{2}} \big)}\frac{e^{-\frac{M^{2}r_{k}^{2b}}{2}}}{\sqrt{2\pi}},
\end{align*}
as $n \to +\infty$ uniformly for $u_{1} \in \{z \in \mathbb{C}: |z-x_{1}|\leq \delta\},\ldots,u_{p} \in \{z \in \mathbb{C}: |z-x_{p}|\leq \delta\}$.
\end{lemma}
\begin{proof}
For convenience, we will write $M_{j}$ and $\lambda_{j}$ instead of $M_{j,k}$ and $\lambda_{j,k}$. By \eqref{asymp prelim of S2kpvp} and Lemma \ref{lemma: uniform}, we have
\begin{align*}
& S_{2k}^{(2)} = \sum_{j:\lambda_{j}\in I_{2}}  \log \Gamma(\tfrac{j+\alpha}{b}) +  \sum_{j:\lambda_{j}\in I_{2}} \log \bigg( \omega_{k} \bigg[ \frac{1}{2}\mathrm{erfc}\Big(-\eta_{j} \sqrt{a_{j}/2}\Big) - R_{a_{j}}(\eta_{j}) \bigg] +   \Omega_{k+1} \bigg),
\end{align*}
where $\eta_{j} = (\lambda_{j}-1)\sqrt{\frac{2 (\lambda_{j}-1-\ln \lambda_{j})}{(\lambda_{j}-1)^{2}}}$. Recall that for all $j \in \{j:\lambda_{j}\in I_{2}\}$, we have $1-\frac{M}{\sqrt{n}} \leq \lambda_{j} = \frac{bnr_{k}^{2b}}{j+\alpha} \leq 1+\frac{M}{\sqrt{n}}$, and that $-M \leq M_{j} \leq M$. Since 
\begin{align}
& \eta_{j} = (\lambda_{j}-1)\bigg( 1 - \frac{\lambda_{j}-1}{3} + \frac{7}{36}(\lambda_{j}-1)^{2} +\bigO((\lambda_{j}-1)^{3})\bigg) = \frac{M_{j}}{\sqrt{n}} - \frac{M_{j}^{2}}{3n} + \frac{7M_{j}^{3}}{36n^{3/2}} + \bigO\bigg(\frac{M^{4}}{n^{2}}\bigg), \nonumber \\
& -\eta_{j} \sqrt{a_{j}/2} = - \frac{M_{j} r_{k}^{b}}{\sqrt{2}} + \frac{5M_{j}^{2} r_{k}^{b}}{6\sqrt{2}\sqrt{n}} - \frac{53 M_{j}^{3} r_{k}^{b}}{72\sqrt{2}n} +\bigO(M^{4}n^{-\frac{3}{2}}), \label{asymp etaj and etajsqrtajover2} 
\end{align}
as $n \to + \infty$ uniformly in $j\in \{j:\lambda_{j}\in I_{2}\}$, using \eqref{asymp of Ra} we obtain
\begin{align*}
& R_{a_{j}}(\eta_{j}) =  \frac{e^{-\frac{M_{j}^{2}r_{k}^{2b}}{2}}}{\sqrt{2\pi}}\bigg( \frac{-1}{3 r_{k}^{b}\sqrt{n}} - \bigg[\frac{M_{j}}{12 r_{k}^{b}} + \frac{5M_{j}^{3} r_{k}^{b}}{18} \bigg]\frac{1}{n} +\bigO((1+M_{j}^{6})n^{-\frac{3}{2}}) \bigg), \\
& \frac{1}{2}\mathrm{erfc}\Big(-\eta_{j} \sqrt{a_{j}/2}\Big) = \frac{1}{2}\mathrm{erfc}\Big(-\frac{M_{j}r_{k}^{b}}{\sqrt{2}}\Big)+\frac{1}{2}\mathrm{erfc}'\Big(-\frac{M_{j} r_{k}^{b}}{\sqrt{2}}\Big)\frac{5M_{j}^{2} r_{k}^{b}}{6\sqrt{2}\sqrt{n}} \\
& + \frac{M_{j}^{3}}{288n}\bigg( 25 M_{j} r_{k}^{2b} \mathrm{erfc}''\Big( -\frac{M_{j} r_{k}^{b}}{\sqrt{2}} \Big) -53 \sqrt{2} r_{k}^{b}\mathrm{erfc}'\Big( -\frac{M_{j} r_{k}^{b}}{\sqrt{2}} \Big) \bigg) + \bigO\Big(e^{-\frac{M_{j}^{2}r_{k}^{2b}}{2}}(1+M_{j}^{8})n^{-\frac{3}{2}}\Big),
\end{align*}
as $n \to + \infty$. Thus we have
\begin{align}
S_{2k}^{(2)} & = \sum_{j:\lambda_{j}\in I_{2}}  \log \Gamma(\tfrac{j+\alpha}{b}) + \sum_{j=g_{k,-}}^{g_{k,+}} \log \bigg\{ \omega_{k} \bigg[ \frac{1}{2}\mathrm{erfc}\Big(-\frac{M_{j} r_{k}^{b}}{\sqrt{2}}\Big)+\frac{1}{2}\mathrm{erfc}'\Big(-\frac{M_{j} r_{k}^{b}}{\sqrt{2}}\Big)\frac{5M_{j}^{2}r_{k}^{b}}{6\sqrt{2}\sqrt{n}} \nonumber \\
& + \frac{M_{j}^{3}}{288n}\bigg( 25 M_{j} r_{k}^{2b} \mathrm{erfc}''\Big( -\frac{M_{j} r_{k}^{b}}{\sqrt{2}} \Big) -53 \sqrt{2} r_{k}^{b}\mathrm{erfc}'\Big( -\frac{M_{j} r_{k}^{b}}{\sqrt{2}} \Big) \bigg) \nonumber \\
& +\frac{e^{-\frac{M_{j}^{2} r_{k}^{2b}}{2}}}{\sqrt{2\pi}}\bigg( \frac{1}{3 r_{k}^{b}\sqrt{n}} + \bigg[\frac{M_{j}}{12 r_{k}^{b}} + \frac{5M_{j}^{3} r_{k}^{b}}{18} \bigg]\frac{1}{n} \bigg) \bigg] +   \Omega_{k+1} \bigg) +\bigO(n^{-1}) \nonumber \\
& = \Sigma_{1}^{(n)}+\frac{1}{\sqrt{n}}\Sigma_{2}^{(n)}+\frac{1}{n}\Sigma_{3}^{(n)} + \bigO(n^{-1}), \label{asymp of S2kp2p in proof}
\end{align}
as $n \to + \infty$, where
\begin{align*}
& \Sigma_{1}^{(n)} = \sum_{j=g_{k,-}}^{g_{k,+}} \log \bigg(  \frac{\omega_{k}}{2}\mathrm{erfc}\Big(-\frac{M_{j} r_{k}^{b}}{\sqrt{2}}\Big) + \Omega_{k+1}\bigg), \\
& \Sigma_{2}^{(n)} = \sum_{j=g_{k,-}}^{g_{k,+}} \frac{\omega_{k}}{\frac{\omega_{k}}{2}\mathrm{erfc}(-\frac{M_{j}r_{k}^{b}}{\sqrt{2}})+\Omega_{k+1}} \bigg( \frac{1}{2}\mathrm{erfc}'\Big( -\frac{M_{j}r_{k}^{b}}{\sqrt{2}} \Big) \frac{5M_{j}^{2}r_{k}^{b}}{6\sqrt{2}} + \frac{e^{-\frac{M_{j}^{2}r_{k}^{2b}}{2}}}{\sqrt{2\pi}}\frac{1}{3r_{k}^{b}} \bigg), \\
& \Sigma_{3}^{(n)} = \sum_{j=g_{k,-}}^{g_{k,+}} \bigg\{ \frac{\omega_{k}}{\frac{\omega_{k}}{2}\mathrm{erfc}(-\frac{M_{j}r_{k}^{b}}{\sqrt{2}})+\Omega_{k+1}} \bigg( \frac{M_{j}^{3}}{288}\bigg[ 25M_{j}r_{k}^{2b}\mathrm{erfc}''\Big( - \frac{M_{j}r_{k}^{b}}{\sqrt{2}} \Big) - 53 \sqrt{2} r_{k}^{b}\mathrm{erfc}'\Big( - \frac{M_{j}r_{k}^{b}}{\sqrt{2}} \Big) \bigg] \\
& + \frac{e^{-\frac{M_{j}^{2}r_{k}^{2b}}{2}}}{\sqrt{2\pi}}\bigg[ \frac{M_{j}}{12r_{k}^{b}} + \frac{5M_{j}^{3}r_{k}^{b}}{18} \bigg] \bigg) - \frac{1}{2} \bigg[ \frac{\omega_{k}}{\frac{\omega_{k}}{2}\mathrm{erfc}(-\frac{M_{j}r_{k}^{b}}{\sqrt{2}})+\Omega_{k+1}} \bigg( \frac{1}{2}\mathrm{erfc}'\Big( -\frac{M_{j}r_{k}^{b}}{\sqrt{2}} \Big) \frac{5M_{j}^{2}r_{k}^{b}}{6\sqrt{2}} + \frac{e^{-\frac{M_{j}^{2}r_{k}^{2b}}{2}}}{\sqrt{2\pi}}\frac{1}{3r_{k}^{b}} \bigg) \bigg]^{2} \bigg\}.
\end{align*}
By \eqref{sum f asymp 2}, the large $n$ asymptotics of $\Sigma_{1}^{(n)}$, $\Sigma_{2}^{(n)}$ and $\Sigma_{3}^{(n)}$ are of the form
\begin{align*}
& \Sigma_{1}^{(n)} = \Sigma_{1,2} \sqrt{n} + \Sigma_{1,3} + \frac{1}{\sqrt{n}} \Sigma_{1,4} + \bigO(M^{4}n^{-1}), \\
& \frac{1}{\sqrt{n}}\Sigma_{2}^{(n)} = \Sigma_{2,3} + \frac{1}{\sqrt{n}} \Sigma_{2,4} + \bigO(n^{-1}), \qquad \frac{1}{n}\Sigma_{3}^{(n)} = \frac{1}{\sqrt{n}} \Sigma_{3,4} + \bigO(n^{-1}),
\end{align*}
where the coefficients $\Sigma_{1,2},\Sigma_{1,3},\Sigma_{1,4},\Sigma_{2,3},\Sigma_{2,4},\Sigma_{3,4}$ are explicit (but we do not write them down). After simplification (using for example $\mathrm{erfc}'(x) = -\frac{2e^{-x^{2}}}{\sqrt{\pi}}$), we obtain
\begin{align*}
& \Sigma_{1,2} =  \widetilde{C}_{2,k}^{(M)}, & & \Sigma_{1,3} + \Sigma_{2,3} = \widetilde{C}_{3,k}^{(n,M)}, & & \Sigma_{1,4} + \Sigma_{2,4} + \Sigma_{3,4} = \widetilde{C}_{4,k}^{(n,M)}.
\end{align*}
\end{proof}
We are now in a position to compute the large $n$ asymptotics of $S_{2k}$ for $k \in \{1,\ldots,m\}$.
\begin{lemma}\label{lemma: asymp of S2k final}
For any $k \in \{1,\ldots,m\}$ and any $x_{1},\ldots,x_{p} \in \mathbb{R}$, there exists $\delta > 0$ such that
\begin{align*}
& S_{2k} = \sum_{j= j_{k,-}}^{j_{k,+}} \log \Gamma(\tfrac{j+\alpha}{b}) + \Big(j_{k,+}-br_{k}^{2b} n \Big) \log \Omega_{k+1}  \\
& + \Big( br_{k}^{2b}n - j_{k,-} \Big) \log  \Omega_{k}  + C_{2,k} \sqrt{n} + C_{3,k} + \frac{1}{\sqrt{n}} C_{4,k} + \bigO(M^{4}n^{-1}),
\end{align*}
as $n \to +\infty$ uniformly for $u_{1} \in \{z \in \mathbb{C}: |z-x_{1}|\leq \delta\},\ldots,u_{p} \in \{z \in \mathbb{C}: |z-x_{p}|\leq \delta\}$, where
\begin{align*}
& C_{2,k} = b r_{k}^{2b} \bigg[ \int_{0}^{\infty} \log \bigg( 1+\frac{\omega_{k}}{2 \Omega_{k+1}}\mathrm{erfc}\frac{t \, r_{k}^{b}}{\sqrt{2}} \bigg) dt + \int_{0}^{\infty} \log \bigg( 1-\frac{\omega_{k}}{2 \Omega_{k}}\mathrm{erfc}\frac{t \, r_{k}^{b}}{\sqrt{2}} \bigg) dt \bigg], \\
& C_{3,k} = \bigg( \frac{1}{2}+\alpha\bigg) \log \Omega_{k+1} + \bigg( \frac{1}{2} - \alpha \bigg) \log \Omega_{k} \\
& + 2br_{k}^{2b} \bigg[ \int_{0}^{\infty} t \log \bigg( 1+\frac{\omega_{k}}{2 \Omega_{k+1}}\mathrm{erfc}\frac{t \, r_{k}^{b}}{\sqrt{2}} \bigg) dt - \int_{0}^{\infty} t \log \bigg( 1-\frac{\omega_{k}}{2 \Omega_{k}}\mathrm{erfc}\frac{t \, r_{k}^{b}}{\sqrt{2}} \bigg) dt \bigg] \\
& + b r_{k}^{2b} \int_{-\infty}^{\infty} \frac{\omega_{k}}{\Omega_{k+1}+\frac{\omega_{k}}{2}\mathrm{erfc}\big( \frac{t\, r_{k}^{b}}{\sqrt{2}} \big)}\bigg[\frac{1}{3r_{k}^{b}}-\frac{5t^{2}r_{k}^{b}}{6}\bigg]\frac{e^{-\frac{t^{2}r_{k}^{2b}}{2}}}{\sqrt{2\pi}}dt, \\
& C_{4,k} = 3br_{k}^{2b} \bigg[ \int_{0}^{\infty} t^{2} \log \bigg( 1+\frac{\omega_{k}}{2 \Omega_{k+1}}\mathrm{erfc}\frac{t \, r_{k}^{b}}{\sqrt{2}} \bigg) dt + \int_{0}^{\infty} t^{2} \log \bigg( 1-\frac{\omega_{k}}{2 \Omega_{k}}\mathrm{erfc}\frac{t \, r_{k}^{b}}{\sqrt{2}} \bigg) dt \bigg] \\
& + b r_{k}^{b} \int_{-\infty}^{\infty} \frac{\omega_{k}}{\Omega_{k+1}+\frac{\omega_{k}}{2}\mathrm{erfc}\big( \frac{t\, r_{k}^{b}}{\sqrt{2}} \big)} \frac{e^{-\frac{t^{2}r_{k}^{2b}}{2}}}{\sqrt{2\pi}} \frac{t(42-193 r_{k}^{2b} t^{2} + 25 r_{k}^{4b}t^{4})}{72} dt \\
& -\frac{1}{2}b r_{k}^{2b} \int_{-\infty}^{\infty} \Bigg( \frac{\omega_{k}}{\Omega_{k+1}+\frac{\omega_{k}}{2}\mathrm{erfc}\big( \frac{t\, r_{k}^{b}}{\sqrt{2}} \big)}\bigg[\frac{1}{3r_{k}^{b}}-\frac{5t^{2}r_{k}^{b}}{6}\bigg]\frac{e^{-\frac{t^{2}r_{k}^{2b}}{2}}}{\sqrt{2\pi}} \Bigg)^{2}dt.
\end{align*}
\end{lemma}
\begin{proof}
By combining Lemmas \ref{lemma:S2kp1p}, \ref{lemma:S2kp3p} and \ref{lemma:S2kp2p}, we have
\begin{align*}
& S_{2k} = \sum_{j= j_{k,-}}^{j_{k,+}} \log \Gamma(\tfrac{j+\alpha}{b}) + \Big(j_{k,+}-br_{k}^{2b} n \Big) \log \Omega_{k+1} + \Big( br_{k}^{2b}n - j_{k,-} \Big) \log  \Omega_{k}  \\
&  + C_{2,k}^{(M)} \sqrt{n} + C_{3,k}^{(n,M)} + \frac{1}{\sqrt{n}} C_{4,k}^{(n,M)} + \bigO(M^{4}n^{-1}),
\end{align*}
as $n \to +\infty$ uniformly for $u_{1} \in \{z \in \mathbb{C}: |z-x_{1}|\leq \delta\},\ldots,u_{p} \in \{z \in \mathbb{C}: |z-x_{p}|\leq \delta\}$, where
\begin{align*}
& C_{2,k}^{(M)} := \widetilde{C}_{2,k}^{(M)} - bM r_{k}^{2b} \bigg[ \log \Omega_{k+1} + \log \Omega_{k} \bigg], \\
& C_{3,k}^{(n,M)} := \widetilde{C}_{3,k}^{(n,M)} + \Big(- bM^{2}r_{k}^{2b}  +\alpha+\theta_{k,+}^{(n,M)} \Big)  \log \Omega_{k+1} + \Big( bM^{2}r_{k}^{2b} -\alpha+\theta_{k,-}^{(n,M)}  \Big) \log \Omega_{k}, \\
& C_{4,k}^{(n,M)} := \widetilde{C}_{4,k}^{(n,M)} - bM^{3}r_{k}^{2b} \bigg[  \log \Omega_{k+1} + \log \Omega_{k} \bigg].
\end{align*}
By choosing $M'$ sufficiently large, we obtain
\begin{align*}
C_{2,k}^{(M)} = C_{2,k} + \bigO(n^{-10}), \qquad C_{3,k}^{(n,M)} = C_{3,k} + \bigO(n^{-10}), \qquad C_{4,k}^{(n,M)} = C_{4,k} + \bigO(n^{-10}),
\end{align*}
as $n \to + \infty$, which finishes the proof.
\end{proof}

We now turn our attention to the asymptotic analysis of $S_{2m+2}^{(2)}$ and then of $S_{2m+2}$.

\begin{lemma}\label{lemma:S2mp2p2p}
For any $x_{1},\ldots,x_{p} \in \mathbb{R}$, there exists $\delta > 0$ such that
\begin{align*}
&  S_{2m+2}^{(2)} = \sum_{j=g_{m+1,-}}^{n} \log \Gamma(\tfrac{j+\alpha}{b}) + \widetilde{C}_{2,m+1}^{(M)}\sqrt{n} + \widetilde{C}_{3,m+1}^{(n,M)} + \frac{1}{\sqrt{n}}\widetilde{C}_{4,m+1}^{(n,M)} + \bigO(M^{4}n^{-1}), \\
& \widetilde{C}_{2,m+1}^{(M)} = \int_{-M}^{-\mathcal{R}}\log \bigg( \frac{\omega_{m+1}}{2}\mathrm{erfc}\bigg( \frac{t}{\sqrt{2b}} \bigg) + \Omega_{m+2} \bigg) dt,  \\
& \widetilde{C}_{3,m+1}^{(n,M)} = \int_{-M}^{-\mathcal{R}} (2t+\mathcal{R}) \log \bigg( \frac{\omega_{m+1}}{2}\mathrm{erfc}\bigg( \frac{t}{\sqrt{2b}} \bigg) + \Omega_{m+2} \bigg) dt \\
& + \bigg( \frac{1}{2}-\theta_{m+1,-}^{(n,M)} \bigg)\log \bigg( \frac{\omega_{m+1}}{2}\mathrm{erfc}\bigg( -\frac{M}{\sqrt{2b}} \bigg) + \Omega_{m+2} \bigg)+ \bigg( \frac{1}{2}+\alpha \bigg)\log \bigg( \frac{\omega_{m+1}}{2}\mathrm{erfc}\bigg( \frac{-\mathcal{R}}{\sqrt{2b}} \bigg) + \Omega_{m+2} \bigg) \\
& + \int_{-M}^{-\mathcal{R}} \frac{\omega_{m+1}}{\Omega_{m+2}+\frac{\omega_{m+1}}{2}\mathrm{erfc}\big( \frac{t}{\sqrt{2b}} \big)}\frac{2b-3\mathcal{R} t -5t^{2}}{6\sqrt{b}}\frac{e^{-\frac{t^{2}}{2b}}}{\sqrt{2\pi}}dt, \\
& \widetilde{C}_{4,m+1}^{(n,M)} = \int_{-M}^{-\mathcal{R}}t(3t+2\mathcal{R}) \log \bigg( \frac{\omega_{m+1}}{2}\mathrm{erfc}\bigg( \frac{t}{\sqrt{2b}} \bigg) + \Omega_{m+2} \bigg) dt \\
& + \int_{-M}^{-\mathcal{R}} \frac{\omega_{m+1}}{\Omega_{m+2}+\frac{\omega_{m+1}}{2}\mathrm{erfc}\big( \frac{t}{\sqrt{2b}} \big)}\frac{e^{-\frac{t^{2}}{2b}}}{\sqrt{2\pi}} \frac{t(42b^{2}-193 b t^{2}+25 t^{4})+6\mathcal{R}(2b^{2}-29b t^{2}+5t^{4})-9 \mathcal{R}^{2}(3 b t -t^{3})}{72b^{3/2}}dt \\
& - \frac{1}{2} \int_{-M}^{-\mathcal{R}} \bigg[ \frac{\omega_{m+1}}{\Omega_{m+2}+\frac{\omega_{m+1}}{2}\mathrm{erfc}\big( \frac{t}{\sqrt{2b}} \big)}\frac{e^{-\frac{t^{2}}{2b}}}{\sqrt{2\pi}} \frac{2b-3\mathcal{R}t-5t^{2}}{6\sqrt{b}} \bigg]^{2}dt \\
& + \bigg[ \frac{1}{12}+\frac{\theta_{m+1,-}^{(n,M)}(\theta_{m+1,-}^{(n,M)}-1)}{2} +\bigg( \frac{1}{2}-\theta_{m+1,-}^{(n,M)} \bigg) \frac{2b+3\mathcal{R}M-5M^{2}}{6}  \bigg]  \frac{\omega_{m+1}}{\Omega_{m+2}+\frac{\omega_{m+1}}{2}\mathrm{erfc}\big( -\frac{M}{\sqrt{2b}} \big)}\frac{e^{-\frac{M^{2}}{2b}}}{\sqrt{2\pi b}} \\
& + \bigg( \bigg( \frac{1}{2}+\alpha \bigg) \frac{b-\mathcal{R}^{2}}{3} - \frac{1+6\alpha+6\alpha^{2}}{12} \bigg) \frac{\omega_{m+1}}{\Omega_{m+2}+\frac{\omega_{m+1}}{2}\mathrm{erfc}\big( -\frac{\mathcal{R}}{\sqrt{2b}} \big)}\frac{e^{-\frac{\mathcal{R}^{2}}{2b}}}{\sqrt{2\pi b}},
\end{align*}
as $n \to +\infty$ uniformly for $u_{1} \in \{z \in \mathbb{C}: |z-x_{1}|\leq \delta\},\ldots,u_{p} \in \{z \in \mathbb{C}: |z-x_{p}|\leq \delta\}$.
\end{lemma}
\begin{proof}
The first part of the proof is identical to the beginning of the proof of Lemma \ref{lemma:S2kp2p}. Namely, in the same way as for \eqref{asymp of S2kp2p in proof}, we have
\begin{align*}
S_{2m+2}^{(2)} & = \sum_{j:\lambda_{j}\in I_{2'}}  \log \Gamma(\tfrac{j+\alpha}{b}) +  \sum_{j=g_{m+1,-}}^{n} \hspace{-0.3cm} \log \bigg( \omega_{m+1} \bigg[ \frac{1}{2}\mathrm{erfc}\Big(-\frac{M_{j} r_{m+1}^{b}}{\sqrt{2}}\Big)+\frac{1}{2}\mathrm{erfc}'\Big(-\frac{M_{j} r_{m+1}^{b}}{\sqrt{2}}\Big)\frac{5M_{j}^{2}r_{m+1}}{6\sqrt{2}\sqrt{n}} \\
& + \frac{M_{j}^{3}}{288n}\bigg( 25 M_{j} r_{m+1}^{2b} \mathrm{erfc}''\Big( -\frac{M_{j} r_{m+1}^{b}}{\sqrt{2}} \Big) -53 \sqrt{2} r_{m+1}^{b}\mathrm{erfc}'\Big( -\frac{M_{j} r_{m+1}^{b}}{\sqrt{2}} \Big) \bigg) \\
& +\frac{e^{-\frac{M_{j}^{2} r_{m+1}^{2b}}{2}}}{\sqrt{2\pi}}\bigg( \frac{1}{3 r_{m+1}^{b}\sqrt{n}} - \bigg[\frac{M_{j}}{12 r_{m+1}^{b}} + \frac{5M_{j}^{3} r_{m+1}^{b}}{18} \bigg]\frac{1}{n} \bigg] +   \Omega_{m+2} \bigg) +\bigO(n^{-1}) \\
& = \sum_{j:\lambda_{j}\in I_{2'}}  \log \Gamma(\tfrac{j+\alpha}{b}) + \Sigma_{1}^{(n)}+\frac{1}{\sqrt{n}}\Sigma_{2}^{(n)}+\frac{1}{n}\Sigma_{3}^{(n)} + \bigO(n^{-1}),
\end{align*}
where
\begin{align*}
& \Sigma_{1}^{(n)} = \sum_{j=g_{m+1,-}}^{n} \log \bigg(  \frac{\omega_{m+1}}{2}\mathrm{erfc}\Big(-\frac{M_{j} r_{m+1}^{b}}{\sqrt{2}}\Big) + \Omega_{m+2}\bigg), \\
& \Sigma_{2}^{(n)} = \sum_{j=g_{m+1,-}}^{n} \frac{\omega_{m+1}}{\frac{\omega_{m+1}}{2}\mathrm{erfc}(-\frac{M_{j}r_{m+1}^{b}}{\sqrt{2}})+\Omega_{m+2}} \bigg( \frac{1}{2}\mathrm{erfc}'\Big( -\frac{M_{j}r_{m+1}^{b}}{\sqrt{2}} \Big) \frac{5M_{j}^{2}r_{m+1}^{b}}{6\sqrt{2}} + \frac{e^{-\frac{M_{j}^{2}r_{m+1}^{2b}}{2}}}{\sqrt{2\pi}}\frac{1}{3r_{m+1}^{b}} \bigg), \\
& \Sigma_{3}^{(n)} = \sum_{j=g_{m+1,-}}^{n} \bigg\{ \frac{\omega_{m+1}}{\frac{\omega_{m+1}}{2}\mathrm{erfc}(-\frac{M_{j}r_{m+1}^{b}}{\sqrt{2}})+\Omega_{m+2}} \bigg( \frac{M_{j}^{3}}{288}\bigg[ 25M_{j}r_{m+1}^{2b}\mathrm{erfc}''\Big( - \frac{M_{j}r_{m+1}^{b}}{\sqrt{2}} \Big) \\
& - 53 \sqrt{2} r_{m+1}^{b}\mathrm{erfc}'\Big( - \frac{M_{j}r_{m+1}^{b}}{\sqrt{2}} \Big) \bigg] + \frac{e^{-\frac{M_{j}^{2}r_{m+1}^{2b}}{2}}}{\sqrt{2\pi}}\bigg[ \frac{M_{j}}{12r_{m+1}^{b}} + \frac{5M_{j}^{3}r_{m+1}^{b}}{18} \bigg] \bigg) \\
& - \frac{1}{2} \bigg[ \frac{\omega_{m+1}}{\frac{\omega_{m+1}}{2}\mathrm{erfc}(-\frac{M_{j}r_{m+1}^{b}}{\sqrt{2}})+\Omega_{m+2}} \bigg( \frac{1}{2}\mathrm{erfc}'\Big( -\frac{M_{j}r_{m+1}^{b}}{\sqrt{2}} \Big) \frac{5M_{j}^{2}r_{m+1}^{b}}{6\sqrt{2}} + \frac{e^{-\frac{M_{j}^{2}r_{m+1}^{2b}}{2}}}{\sqrt{2\pi}}\frac{1}{3r_{m+1}^{b}} \bigg) \bigg]^{2} \bigg\},
\end{align*}
where, for conciseness, we have written $M_{j}$ instead of $M_{j,m+1}$. Note that one cannot yet apply Lemma \ref{lemma:Riemann sum} because $r_{m+1}$ depends on $n$ (if $\mathcal{R} \neq 0$). Using $\mathrm{erfc}'(x)=-\frac{2e^{-x^{2}}}{\sqrt{\pi}}$, $\mathrm{erfc}''(x) = \frac{4xe^{-x^{2}}}{\sqrt{\pi}}$ and
\begin{align*}
r_{m+1}^{b} = \frac{1}{\sqrt{b}} + \frac{\mathcal{R}}{2\sqrt{b}}\frac{1}{\sqrt{n}} - \frac{\mathcal{R}^{2}}{8\sqrt{b}n} + \bigO(n^{-\frac{3}{2}}), \qquad \mbox{as } n \to + \infty,
\end{align*}
we obtain
\begin{align*}
\Sigma_{1}^{(n)}+\frac{1}{\sqrt{n}}\Sigma_{2}^{(n)}+\frac{1}{n}\Sigma_{3}^{(n)} = \widetilde{\Sigma}_{1}^{(n)}+\frac{1}{\sqrt{n}}\widetilde{\Sigma}_{2}^{(n)}+\frac{1}{n}\widetilde{\Sigma}_{3}^{(n)} + \bigO(n^{-1}), \qquad \mbox{as } n \to + \infty,
\end{align*}
where
\begin{align*}
& \widetilde{\Sigma}_{1}^{(n)} = \sum_{j=g_{m+1,-}}^{n} \log \bigg(  \frac{\omega_{m+1}}{2}\mathrm{erfc}\Big(-\frac{M_{j}}{\sqrt{2b}}\Big) + \Omega_{m+2}\bigg), \\
& \widetilde{\Sigma}_{2}^{(n)} = \sum_{j=g_{m+1,-}}^{n} \frac{\omega_{m+1}}{\frac{\omega_{m+1}}{2}\mathrm{erfc}(-\frac{M_{j}}{\sqrt{2b}})+\Omega_{m+2}}\frac{1}{2\sqrt{b}}\bigg(\frac{2b-5M_{j}^{2}}{3}+\mathcal{R}M_{j}\bigg)\frac{e^{-\frac{M_{j}^{2}}{2b}}}{\sqrt{2\pi}}, \\
& \widetilde{\Sigma}_{2}^{(n)} = \sum_{j=g_{m+1,-}}^{n} \Bigg\{ \frac{\omega_{m+1}}{\frac{\omega_{m+1}}{2}\mathrm{erfc}(-\frac{M_{j}}{\sqrt{2b}})+\Omega_{m+2}} \bigg( \frac{6b^{2}M_{j}+73bM_{j}^{3}-25M_{j}^{5}-(12 b^{2} + 42 b M_{j}^{2}-30M_{j}^{4})\mathcal{R}}{72 b^{3/2}} \\ 
&- \frac{9(b M_{j}+M_{j}^{3})\mathcal{R}^{2}}{72b^{3/2}} \bigg) \frac{e^{-\frac{M_{j}^{2}}{2b}}}{\sqrt{2\pi}}  - \frac{\omega_{m+1}^{2}}{(\frac{\omega_{m+1}}{2}\mathrm{erfc}(-\frac{M_{j}}{\sqrt{2b}})+\Omega_{m+2})^{2}}\frac{(2b+3M_{j}\mathcal{R}-5M_{j}^{2})^{2}}{72b}\frac{e^{-\frac{M_{j}^{2}}{b}}}{2\pi} \Bigg\}.
\end{align*}
 The large $n$ asymptotics of $\widetilde{\Sigma}_{1}^{(n)}$, $\widetilde{\Sigma}_{2}^{(n)}$ and $\widetilde{\Sigma}_{3}^{(n)}$ can now be obtained using \eqref{sum f asymp 2 edge}, and by choosing $M'$ sufficiently large, we obtain the claim after a computation and rearranging the terms.
\end{proof}

\begin{lemma}\label{lemma: asymp of S2mp2 final}
For any $x_{1},\ldots,x_{p} \in \mathbb{R}$, there exists $\delta > 0$ such that
\begin{align*}
& S_{2m+2} = \sum_{j= j_{m+1,-}}^{n} \log \Gamma(\tfrac{j+\alpha}{b}) + \Big( n - j_{m+1,-} \Big) \log  \Omega_{m+1} \\
&  + C_{2,m+1} \sqrt{n} + C_{3,m+1} + \frac{1}{\sqrt{n}} C_{4,m+1} + \bigO(M^{4}n^{-1}),
\end{align*}
as $n \to +\infty$ uniformly for $u_{1} \in \{z \in \mathbb{C}: |z-x_{1}|\leq \delta\},\ldots,u_{p} \in \{z \in \mathbb{C}: |z-x_{p}|\leq \delta\}$, where
\begin{align*}
& C_{2,m+1} = \int_{0}^{\infty} \log \bigg( 1-\frac{\omega_{m+1}}{2 \Omega_{m+1}}\mathrm{erfc}\frac{t}{\sqrt{2b}} \bigg) dt + \mathcal{R}\log \frac{\Omega_{m+1}}{\Omega_{m+2}} + \int_{0}^{-\mathcal{R}} \log \bigg( 1+\frac{\omega_{m+1}}{2\Omega_{m+2}}\mathrm{erfc} \frac{t}{\sqrt{2b}} \bigg)dt, \\
& C_{3,m+1} = \bigg( \frac{1}{2}+\alpha\bigg) \log \bigg(  \Omega_{m+2} + \frac{\omega_{m+1}}{2}\mathrm{erfc} \frac{-\mathcal{R}}{\sqrt{2b}} \bigg) + \bigg( \frac{1}{2} - \alpha \bigg) \log \Omega_{m+1} \\
& - \int_{0}^{\infty} (2t-\mathcal{R}) \log \bigg( 1-\frac{\omega_{m+1}}{2 \Omega_{m+1}}\mathrm{erfc}\frac{t}{\sqrt{2b}} \bigg) dt + \int_{0}^{-\mathcal{R}} (2t+\mathcal{R}) \log \bigg( 1+\frac{\omega_{m+1}}{2 \Omega_{m+2}}\mathrm{erfc}\frac{t}{\sqrt{2b}} \bigg) dt  \\
& + \int_{-\infty}^{-\mathcal{R}} \frac{\omega_{m+1}}{\Omega_{m+2}+\frac{\omega_{m+1}}{2}\mathrm{erfc}\big( \frac{t}{\sqrt{2b}} \big)} \frac{2b-3\mathcal{R}t-5t^{2}}{6\sqrt{b}} \frac{e^{-\frac{t^{2}}{2b}}}{\sqrt{2\pi}}dt, \\
& C_{4,m+1} = \int_{0}^{\infty} (3t^{2}-2\mathcal{R}t) \log \bigg( 1-\frac{\omega_{m+1}}{2 \Omega_{m+1}}\mathrm{erfc}\frac{t}{\sqrt{2b}} \bigg) dt + \int_{0}^{-\mathcal{R}} (3t^{2}+2\mathcal{R}t) \log \bigg( 1+\frac{\omega_{m+1}}{2 \Omega_{m+2}}\mathrm{erfc}\frac{t}{\sqrt{2b}} \bigg) dt  \\
& + \int_{-\infty}^{-\mathcal{R}} \frac{\omega_{m+1}}{\Omega_{m+2}+\frac{\omega_{m+1}}{2}\mathrm{erfc}\big( \frac{t}{\sqrt{2b}} \big)} \frac{e^{-\frac{t^{2}}{2b}}}{\sqrt{2\pi}} \frac{t(42b^{2}-193 b t^{2} + 25 t^{4})+6 \mathcal{R}(2b^{2}-29bt^{2}+5t^{4})-9\mathcal{R}^{2}(3bt-t^{3})}{72b^{3/2}} dt \\
& -\frac{1}{2} \int_{-\infty}^{-\mathcal{R}} \Bigg( \frac{\omega_{m+1}}{\Omega_{m+2}+\frac{\omega_{m+1}}{2}\mathrm{erfc}\big( \frac{t}{\sqrt{2b}} \big)} \frac{2b-3 \mathcal{R}t-5t^{2}}{6\sqrt{b}} \frac{e^{-\frac{t^{2}}{2b}}}{\sqrt{2\pi}} \Bigg)^{2}dt \\
& + \bigg( \bigg(\frac{1}{2}+\alpha\bigg)\frac{b-\mathcal{R}^{2}}{3} - \frac{1+6\alpha+6\alpha^{2}}{12} \bigg) \frac{\omega_{m+1}}{\Omega_{m+2}+\frac{\omega_{m+1}}{2}\mathrm{erfc}\big( -\frac{\mathcal{R}}{\sqrt{2b}} \big)} \frac{e^{-\frac{\mathcal{R}^{2}}{2b}}}{\sqrt{2\pi b}}.
\end{align*}
\end{lemma}
\begin{proof}
By combining Lemmas \ref{lemma:S2kp3p} and \ref{lemma:S2mp2p2p}, we have
\begin{align*}
& S_{2m+2} = \sum_{j= j_{m+1,-}}^{n} \log \Gamma(\tfrac{j+\alpha}{b}) + \Big( n - j_{m+1,-} \Big) \log  \Omega_{m+1}  \\
&  + C_{2,m+1}^{(M)} \sqrt{n} + C_{3,m+1}^{(n,M)} + \frac{1}{\sqrt{n}} C_{4,m+1}^{(n,M)} + \bigO(M^{4}n^{-1}),
\end{align*}
as $n \to +\infty$ uniformly for $u_{1} \in \{z \in \mathbb{C}: |z-x_{1}|\leq \delta\},\ldots,u_{p} \in \{z \in \mathbb{C}: |z-x_{p}|\leq \delta\}$, where
\begin{align*}
& C_{2,m+1}^{(M)} = \widetilde{C}_{2,m+1}^{(M)} + (\mathcal{R}- M)  \log \Omega_{m+1}, \\
& C_{3,m+1}^{(n,M)} = \widetilde{C}_{3,m+1}^{(n,M)} + \Big( M^{2} - M \mathcal{R} -\alpha+\theta_{k,-}^{(n,M)}  \Big)\log \Omega_{m+1}, \\
& C_{4,m+1}^{(n,M)} = \widetilde{C}_{4,m+1}^{(n,M)} + (- M^{3} + M^{2}\mathcal{R}) \log \Omega_{m+1}.
\end{align*}
Note that $S_{2m+2}$ is independent of $M$. By choosing $M'$ sufficiently large, we obtain
\begin{align*}
C_{2,m+1}^{(M)} = C_{2,m+1} + \bigO(n^{-10}), \quad C_{3,m+1}^{(n,M)} = C_{3,m+1} + \bigO(n^{-10}), \quad C_{4,m+1}^{(n,M)} = C_{4,m+1} + \bigO(n^{-10}),
\end{align*}
as $n \to + \infty$, which concludes the proof.
\end{proof}

We now finish the proof of Theorem \ref{thm:main thm}.

\begin{proof}[Proof of Theorem \ref{thm:main thm}]
Combining \eqref{log Dn as a sum of sums} with Lemmas \ref{lemma: S0}, \ref{lemma: S2km1}, \ref{lemma: asymp of S2k final} and \ref{lemma: asymp of S2mp2 final}, we infer that for any $x_{1},\ldots,x_{p} \in \mathbb{R}$, there exists $\delta > 0$ such that
\begin{align*}
& \log D_{n} = S_{-1}+S_{0}+\sum_{k=1}^{m}(S_{2k-1}+S_{2k}) + S_{2m+1} + S_{2m+2}\\
& = -\frac{1}{2b}n^{2}\log n - \frac{1+2\alpha}{2b}n \log n + n \log \frac{\pi}{b} + M' \log \Omega_{1} + \sum_{j=1}^{M'} \log \Gamma (\tfrac{j+\alpha}{b}) \\
& + \sum_{k=1}^{m} \bigg\{ (j_{k,-}-j_{k-1,+}-1) \log \Omega_{k} + \sum_{j=j_{k-1,+}+1}^{j_{k,-}-1}  \hspace{-0.3cm} \log \Gamma(\tfrac{j+\alpha}{b}) + \sum_{j= j_{k,-}}^{j_{k,+}} \log \Gamma(\tfrac{j+\alpha}{b}) \\
& + \Big(j_{k,+}-b r_{k}^{2b} n \Big) \log  \Omega_{k+1} + \Big( br_{k}^{2b}n - j_{k,-} \Big) \log  \Omega_{k} + C_{2,k} \sqrt{n} + C_{3,k} + \frac{1}{\sqrt{n}} C_{4,k} \bigg\} \\
& + (j_{m+1,-}-j_{m,+}-1) \log \Omega_{m+1} + \sum_{j=j_{m,+}+1}^{j_{m+1,-}-1}  \hspace{-0.3cm} \log \Gamma(\tfrac{j+\alpha}{b}) + \sum_{j= j_{m+1,-}}^{n} \log \Gamma(\tfrac{j+\alpha}{b}) \\
& + \Big( n - j_{m+1,-} \Big) \log  \Omega_{m+1}  + C_{2,m+1} \sqrt{n} + C_{3,m+1} + \frac{1}{\sqrt{n}} C_{4,m+1} + \bigO(M^{4}n^{-1})
\end{align*}
as $n \to +\infty$ uniformly for $u_{1} \in \{z \in \mathbb{C}: |z-x_{1}|\leq \delta\},\ldots,u_{p} \in \{z \in \mathbb{C}: |z-x_{p}|\leq \delta\}$. Recall that $Z_{n}$ is given by \eqref{explicit formula for Zn}. Hence, simplifying the above yields
\begin{align*}
& \log D_{n} = \log Z_{n} + M \log \Omega_{1} + \sum_{k=1}^{m}\bigg\{ \Big(j_{k,+}-br_{k}^{2b} n \Big) \log \Omega_{k+1} + \Big( br_{k}^{2b}n - j_{k-1,+}-1 \Big) \log \Omega_{k} + C_{2,k} \sqrt{n} \\
&  + C_{3,k} + \frac{1}{\sqrt{n}}C_{4,k} \bigg\} + (n-j_{m,+}-1) \log \Omega_{m+1} + C_{2,m+1} \sqrt{n} + C_{3,m+1} + \frac{1}{\sqrt{n}} C_{4,m+1} + \bigO(M^{4}n^{-1}),
\end{align*}
as $n \to +\infty$. Using \eqref{def of Dn as n fold integral}, we then find
\begin{align*}
& \log \mathbb{E}\bigg[ \prod_{j=1}^{p} e^{u_{j}N(D_{r_{j}})} \bigg] =  (br_{1}^{2b}n-1) \log \Omega_{1} + \sum_{k=2}^{m} \big( br_{k}^{2b}n-br_{k-1}^{2b}n-1 \big) \log \Omega_{k}  \\
& + (n-br_{m}^{2b}n-1)\log \Omega_{m+1} + \sum_{k=1}^{m+1} \Big\{C_{2,k} \sqrt{n} + C_{3,k} + \frac{1}{\sqrt{n}} C_{4,k} \Big\} + \bigO(M^{4}n^{-1}),
\end{align*}
which can be rewritten as \eqref{asymp in main thm} using \eqref{functions Fb and Gb}, \eqref{def of omegaell} and \eqref{def of Omega j}. This finishes the proof of Theorem \ref{thm:main thm}.
\end{proof}

\paragraph{Acknowledgements.} The author is grateful to Peter Forrester, Arno Kuijlaars, Gr\'{e}gory Schehr and Nick Simm for useful comments. The author acknowledges support from the European Research Council, Grant Agreement No. 682537, the Swedish Research Council, Grant No. 2015-05430, the Ruth and Nils-Erik Stenb\"ack Foundation, and the Novo Nordisk Fonden Project, Grant 0064428.

\end{document}